\documentclass[prd,aps,twocolumn,superscriptaddress,amsmath,amssymb,floatfix]{revtex4-1}

\usepackage{datetime}
\usepackage[utf8]{inputenc}

\usepackage[colorlinks=true,citecolor=blue,filecolor=blue,linkcolor=blue,urlcolor=blue]{hyperref}

\usepackage{graphicx}               %
\usepackage[caption=false]{subfig}  %

\usepackage{verbatim}               %
\usepackage{xspace}                 %

\usepackage{enumitem}               %
\usepackage{bm}                     %

\usepackage{amsthm}                %
\newtheorem{proposition}{Proposition}

\usepackage{fontawesome}  %

\usepackage{amsmath}
\usepackage{accents}
\newlength{\dhatheight}
\newcommand{\tildehat}[1]{%
    \settoheight{\dhatheight}{\ensuremath{\hat{#1}}}%
    \addtolength{\dhatheight}{-0.35ex}%
    \tilde{\vphantom{\rule{1pt}{\dhatheight}}%
    \smash{\hat{#1}}}}

\begin{document}

\title{Deficit hawks: robust new physics searches with unknown backgrounds}
\author{Jelle Aalbers}\email{jaalbers@stanford.edu}\affiliation{Kavli Institute for Particle Astrophysics and Cosmology, Stanford University, Stanford, CA 94305, USA}\affiliation{SLAC National Accelerator Laboratory, Menlo Park, CA 94025, USA} 

\usdate
\date{\today}

\begin{abstract}
Searches for new physics often face unknown backgrounds, causing false detections or weakened upper limits.
This paper introduces the \emph{deficit hawk} technique, which mitigates unknown backgrounds by testing multiple options for data cuts, such as fiducial volumes or energy thresholds.
Combining the power of likelihood ratios with the robustness of the interval-searching techniques, deficit hawks could improve mean upper limits on new physics by a factor two for experiments with partial or speculative background knowledge.
Deficit hawks are well-suited to analyses that use machine learning or other multidimensional discrimination techniques, and can be extended to permit discoveries in regions without unknown background.~\href{https://github.com/JelleAalbers/deficithawks}{\faicon{github}}
\end{abstract}

\maketitle

\section{Introduction}

\begin{figure*}
    \centering
    \subfloat{{\includegraphics[height=6.46cm]{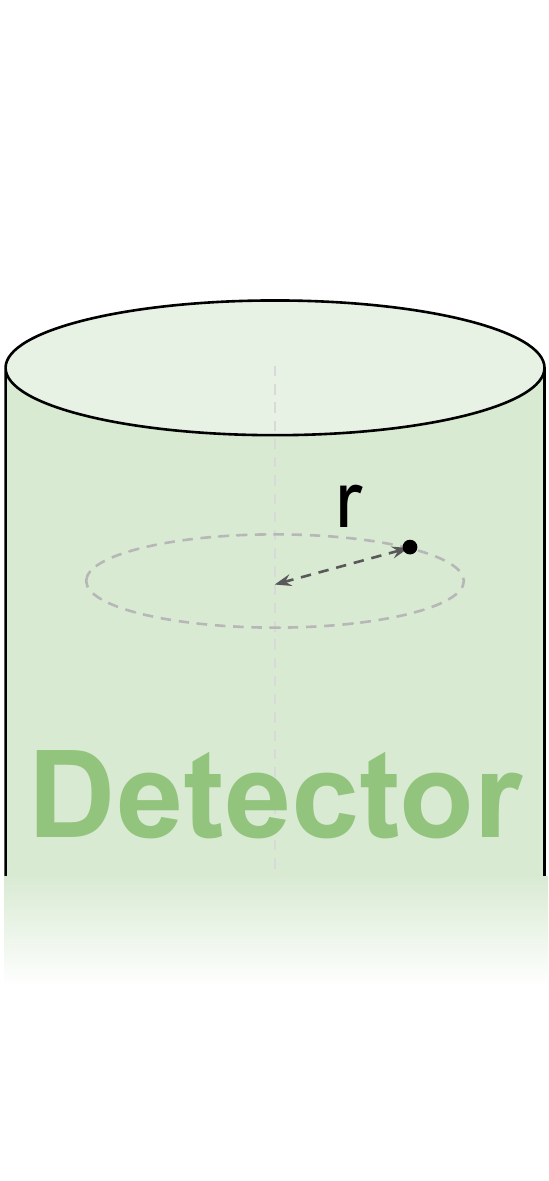}}}%
    \qquad
    \subfloat{{\includegraphics[height=6.46cm]{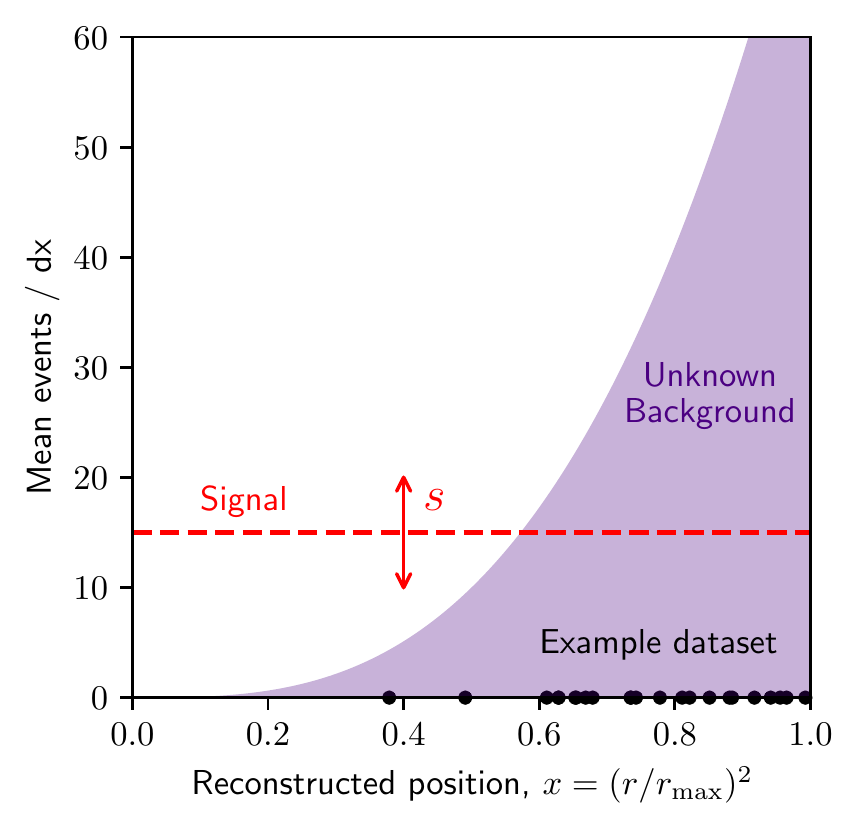} }}%
    \qquad
    \subfloat{{\includegraphics[height=6.46cm]{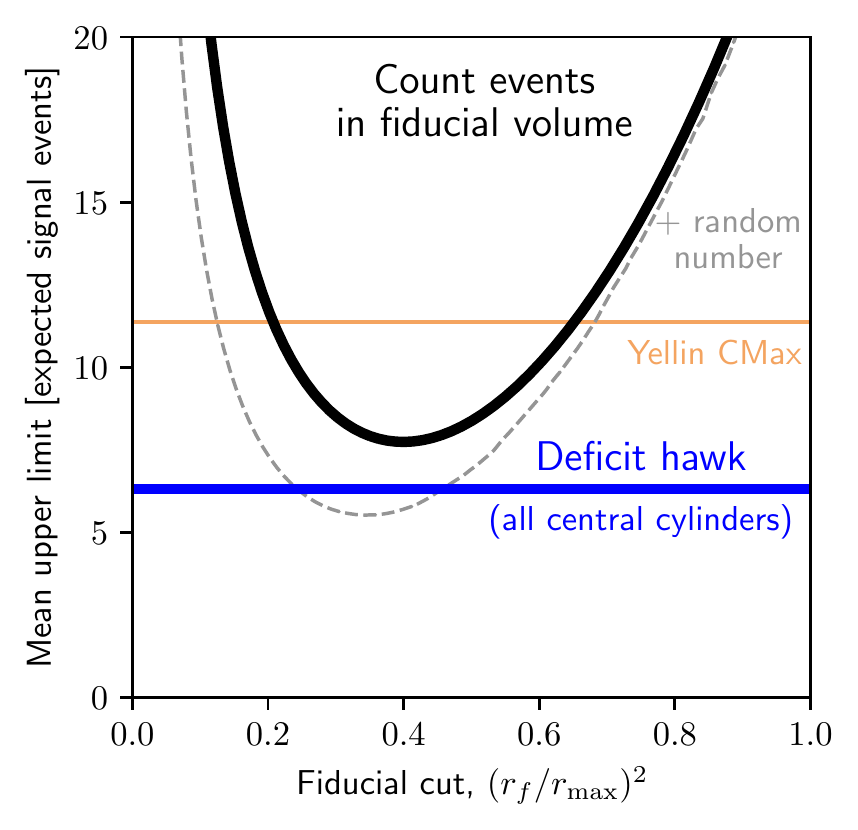} }}%
    \caption{A simple example that illustrates the deficit hawk method. Left: sketch of a long cylindrical detector, in which physicists observe the radial position $r$ of events from particle interactions. Center: A possible observed dataset (black dots), the homogeneous signal model (dashed red) on whose strength $s$ the experiment aims to set upper limits, and a possible unknown background (purple shaded), all as a function of $x = (r/r_\text{max})^2$, with $r_\text{max}$ the detector's radius. Right: Mean upper limits on the signal strength $s$, averaged over datasets generated from the unknown background shown in the middle plot. The solid black line is for counting events inside an inner cylinder with radius $r_f$, chosen before observing the data. The dashed black line is for adding a uniform(0,1) random number to this count. The solid blue line is for a deficit hawk that tests all central cylinders, and the thin orange line is for Yellin's CMax / optimum interval method, which tests all intervals in $x$. The latter methods require no fiducial volume choice, so they are represented by horizontal lines.
    }%
    \label{fig:simple_example}%
\end{figure*}

\subsection{Motivation}

Experiments often aim to constrain signals among partially or completely \emph{unknown backgrounds}. This challenge is common for (in)direct dark matter searches, neutrino experiments, and in astroparticle physics in general. The challenge takes many different forms:
\begin{itemize}
    \item \emph{Difficult regions} with possible extra backgrounds often exist near detector edges, at low energies, or at the start of observing runs. Should these be added to analyses or not?
    \item \emph{Data quality cuts} can remove backgrounds that would be difficult or tedious to model. But without a model, how do we choose good cut thresholds?
    \item \emph{Sideband analyses} may have poorly understood backgrounds and limited power to discriminate them from signals. %
    \item For \emph{new detector technologies}, backgrounds might only be known approximately, or not at all.
\end{itemize}

Established statistical techniques exist for several of these situations. 
With a parametric, but uncertain, background model, analysts can profile or marginalize over nuisance parameters \cite{pdg_review, discrete_profile}.
Without a parametric model, some backgrounds may still be constrained or mitigated using complementary (calibration) measurements \cite{safeguard, algeri}.
Finally, analysts that wish to make \emph{no} assumptions on the background can still set upper limits on new physics signals by counting observed events, or more efficiently, with Yellin's maximum gap or optimum interval methods \cite{optitv, optitv2}. 

However, many analysts find themselves in little-studied intermediate situations. 
Perhaps they suspect a general property of the unknown background -- for example, that it is larger at low energies. Or, they might have a model only for some background components, and know that additional components may exist. Finally, analysts may be comfortable with their background model in some, but not all, of the measured space.

These situations can be shoehorned into existing methods only at a price. 
Regions with too-uncertain backgrounds could be cut, but at a cost to sensitivity. 
Using many nuisance parameters weakens an analysis' physics reach, and still risks incorrect results if the true background model is different from the proposed parametric form(s).
Treating the background as fully unknown loses any chance of making discoveries; and while partial background models can be ignored, used to motivate cuts, or added to the signal model, neither is as effective as likelihood-ratio discrimination.

\subsection{Synopsis}
This work introduces the \emph{deficit hawk} technique. It allows analysts to exploit partial, uncertain, and even informal background knowledge. If unknown backgrounds appear, limits set by deficit hawks remain correct, and are weakened less than limits from standard techniques. Deficit hawks allow the use of likelihood ratios to powerfully discriminate \emph{known} backgrounds from signals. Finally, if a core part of the measurement has \emph{only} known backgrounds, deficit hawks naturally enhance the exclusion power of a two-sided likelihood-ratio analysis done in the core region.

Deficit hawks do not represent a single statistical method. The name is merely a convenient label for an idea -- to use likelihood ratio tests (or similar statistics) to test \emph{multiple regions} of the data, such as different fiducial volume cuts or energy thresholds.
Then, base an upper limit on the region that shows the most deficit-like result, while accounting for the fact that multiple regions were tested, much as in a look-elsewhere correction or a global significance computation. %

This idea is not original: it is the same principle that drives Yellin's optimum interval and maximum gap methods \cite{optitv, optitv2} -- which are examples of deficit hawks, and widely used in particle physics. More generally, statisticians have long studied how to draw valid conclusions from multiple tests \cite{bonferroni}.
This work's contribution is twofold: first, to show that choosing specific deficit hawks that leverage domain knowledge can give stronger limits (by around a factor two) than testing all intervals with a counting-based statistic; and second, to show that deficit hawks apply to situations well beyond one-dimensional experiments with fully unknown backgrounds.

\subsection{Example: avoiding a fiducial volume choice}
\label{sec:simple_example}
This section gives a basic example of how deficit hawks work, and why they are useful.
Figure \ref{fig:simple_example} sketches an experiment that uses a long cylindrical detector, in which scientists reconstruct the radial position $r$ of events from particle interactions. The goal is to set an upper limit on the strength $s$ (in expected number of events) of homogeneously distributed signals, such as rare decays or scatters of a weakly interacting particle. Backgrounds, e.g. from external radioactivity, are stronger near the detector edge, but analysts have no formal background model.

A conventional approach would be to pick an inner central cylinder with radius $r_f$, \emph{before} observing the data. Then, analysts can set an upper limit on $s$ based on the number of events with $r < r_f$ -- that is, they do a Poisson test inside a fiducial volume.
The performance of this method depends on $r_f$, and is shown in black in the right panel of figure \ref{fig:simple_example}, for the unknown background sketched in the center panel.
We will test many different unknown background distributions later in this work.

Clearly, the key challenge is to choose a good $r_f$ relative to the \emph{unknown} background. To do this, experiments might sacrifice and examine some fraction of their data for this purpose; or make an informed guess based on partial models, other observations, or theoretical arguments; or they could pick a round number and hope for the best.

The deficit hawk approach to this problem is to \emph{try many fiducial volumes}. The options should still be chosen before seeing the data. For example, we can try \emph{all possible values} of $r_f$. Clearly, we should not just compute upper limits for each choice, and publish the strongest one -- that would be as bad as choosing a fiducial volume after seeing the data. However, we can define our test statistic as `the most deficit-like observation among the tested volumes' (as formalized in section \ref{sec:deficit_hawks}). Next, we use Monte Carlo simulations of signal-only models to compute the distribution of this statistic, as we might do for any statistic, and use this to set valid upper limits.

Figure \ref{fig:simple_example} shows the mean upper limit from trying all possible $r_f$ values as the blue horizontal line. The deficit hawk performs similar to the optimal fiducial volume choice.
In fact, the deficit hawk performs \emph{better} than the best volume in this situation -- though narrowly, and only for a technical reason. Discrete statistics such as counts inherently perform slightly worse in frequentist limit setting than continuous statistics, as discussed in section \ref{sec:notation} and \cite{augmented}. If we add a random number between 0 and 1 to the count, such a modified Poisson test can beat the deficit hawk (if $r_f$ is chosen optimally), as shown in the dashed black curve in figure \ref{fig:simple_example}.

There is no need to test all $r_f$ values: the experiment might choose to try a few hand-picked fiducial volume options instead. 
But trying all $r_f$ options is simpler than it seems: as discussed in section \ref{sec:region_choice} and appendix \ref{sec:compute}, we need only test the $r_f$ values of observed events (and $r_\text{max}$), as cylinders with `free space' beyond them never give the most deficit-like observation.

Real experiments are more complicated than this example: they may have multiple observables on which they place cuts, or have \emph{known} backgrounds besides the potential unknown background, or may want to allow discovery claims rather than only upper limits. We will show how the deficit hawk method generalizes to these situations later in this work.

\subsection{Outline}

After section \ref{sec:notation} introduces basic concepts and notation, section \ref{sec:deficit_hawks} defines deficit hawks and discusses their basic features.
Section \ref{sec:region_choice} compares different choices of regions to test, and section \ref{sec:stat_choice} examines different statistics for comparing the regions. Next, section \ref{sec:known_bg} considers how to leverage partial background models; section \ref{sec:underflucts} discusses background underfluctuations; and \ref{sec:detections} proposes a way to allow discovery sensitivity if some region is known to be free of unknown background.
Section \ref{sec:discussion} discusses potential objections and limitations, and section \ref{sec:conclusions} concludes with a summary and suggestions for further research.
The appendices prove some properties of likelihood ratios used in the text (\ref{sec:proofs}), discuss computational costs and ways to reduce it (\ref{sec:compute}), and show performance tests for additional background scenarios (\ref{sec:extra_scen}). 

\section{Basic theory}
\subsection{Background and conventions}
\label{sec:notation}
This section reviews likelihood inference techniques used in particle physics, and defines notation used for the rest of this work. For more extensive reviews, see e.g.~references \cite{pdg_review, asymptotic, nature_review, severini}.

Take an experiment that observes a dataset $D$ of $N$ events. 
For each event, the experiment records observed quantities such as energy, time, or reconstructed position -- that is, events are vectors ${\bf{x}} \in X \subseteq \mathbb{R}^m$ for some $m$. The events are produced by Poisson processes, such as radioactive decays. The experiment aims to constrain a parameter $s$, such as a cross-section or decay constant for which positive values indicate new physics. For simplicity, in most examples below, we let $s$ equal the expected number of signal events, $\mu_\text{sig} = s$, but this is not an assumption of the method.

To set frequentist confidence intervals on $s$, analysts must choose a \emph{test statistic} $T(s|D)$, which summarizes the data $D$ in a single number for each possible $s$, and an \emph{interval construction}, which translates (`inverts') the observed $T(s)$ to a confidence interval on $s$. A powerful and commonly used test statistic is the (unsigned) log likelihood ratio \cite{neyman_pearson, pdg_review}:
\begin{equation}
\label{eq:original_plr}
u(s | D) = -2 \ln \frac{L(s| D)}{L(\hat{s}| D)}
\end{equation}
Here, $L$ is the likelihood function of the data and $\hat{s}$ denotes the `best fit', i.e.~the $s$ that maximizes $L(s)$ in the domain of $s$-values that the model physically allows. For example, $L$ might be an (extended) unbinned likelihood:
\begin{equation}
\label{eq:unbinned_l}
\ln L = \ln \mathrm{Pr} - \mu + \sum_{\text{events} \, i} \ln r(\bm{x}_i)
\end{equation}
with $\mathrm{Pr}$ an optional `prior' term that represents external constraints independent of the data, $r(\bm{x}) = \mu \cdot \mathrm{PDF}(\bm{x})$ the differential rate of events at $\bm{x}$, and $\mathrm{PDF}(\bm{x})$ the probability density of events. All of $\mathrm{Pr}$, $\mu$ and $r$ are generally functions of $s$. We will use unbinned likelihoods throughout this work, though binned likelihoods can be used equally well.

If $s$ controls a signal strength, it is useful to distinguish excesses and deficits. We will use a \emph{signed} likelihood ratio:
\begin{equation}
\label{eq:signed_plr}
t(s) = \sqrt{u(s)} \cdot \mathrm{sign}(\mu(\hat{s}) - \mu(s)),
\end{equation}
which is positive for excesses ($\mu(\hat{s}) > \mu(s)$) and negative for deficits. The blue line in figure \ref{fig:neyman} shows an example: like many observations, this observation is an excesses for low $s$ and a deficits for higher $s$.
Many alternate signed likelihood definitions exist that behave similarly in cases of interest here: e.g. omitting the square root and using different sign conventions \cite{knut_fc}, setting the likelihood ratio to 0 for excesses when only upper limits are of interest \cite{asymptotic, pdg_review}, or using $\mathrm{sign}(\hat{s} - s)$ for the sign function \cite{severini}. Under some regularity conditions, and for sufficiently large signals, $t$ is distributed as a standard Gaussian \cite{severini} -- but we will not rely on this asymptotic approximation anywhere in this work.

\begin{figure}
    \centering
    \includegraphics[width=\columnwidth]{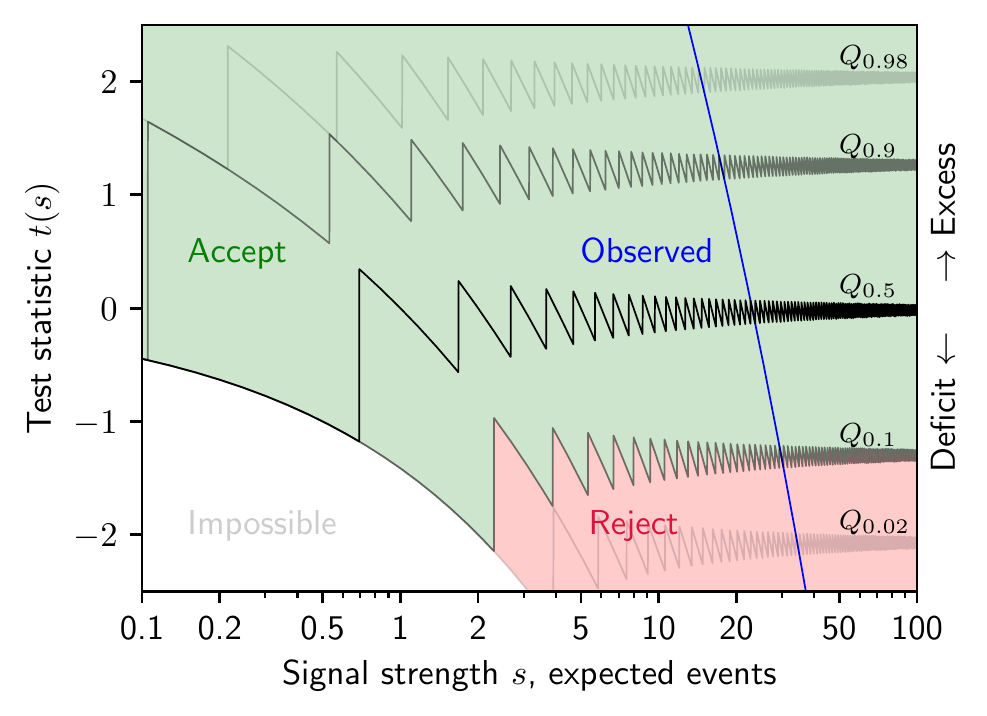}
    \caption{Illustration of setting 90\% confidence upper limits with a signed likelihood ratio, for the example experiment shown in figure \ref{fig:simple_example} (with no fiducial volume cut). The blue line shows the observed value of $t(s)$. Thin black lines show quantiles $Q_\alpha$ of the distribution of $t$, assuming for each $s$ that it is the true signal strength. Rapid jumps are due to the discreteness of $N$. If $t(s)$ falls in the green-shaded `Accept' region, $s$ is included in the confidence interval. For 90\% confidence level upper limits on signal strengths using $t$, the `Accept' region sits at and above the tenth percentile $Q_{0.1}$. Thus, on the observed dataset, we get the limit $s \lesssim 30.5$, as that is the highest $s$ for which $t(s) \geq Q_{0.1}(s)$. Since there are no backgrounds or priors, $t=t_0$, and the test in this example is equivalent to a Poisson test.}
    \label{fig:neyman}
\end{figure}

If there are no priors or known backgrounds, $t(s)$ becomes a simple expression $t_0(s)$ that depends only on $s$ and $N$. More precisely, if we can choose $\mathbf{x}$ coordinates so that $\mathrm{PDF}(\mathbf{x}) = 1$ at all $\mathbf{x}$, then $\ln L = - \mu + N \ln \mu$, $\mu(\hat{s}) = N$, and
\begin{align}
\label{eq:t0}
u(s)/2 &= \; \mu - N + N \ln( N/\mu) \nonumber \\
t(s) = t_0(s) &= \mathrm{sign}(N - \mu(s)) \sqrt{2}\nonumber \\
         &\;\;\;\;\cdot \sqrt{\mu(s) - N + N \ln(N/\mu(s))} ,
\end{align}
with $0 \ln 0$ understood as 0.

An interval construction translates the observed test statistic values into a confidence interval on $s$. Frequentist interval constructions should satisfy \emph{coverage}: for 90\% confidence intervals, the true $s$ must be inside the interval in $\geq90\%$ of repeated experiments. Equivalently, false exclusions of the true $s$ may happen in $\leq 10\%$ of trials. 
In this work, we will say that a method `has coverage' even if it includes the truth \emph{more often} than the confidence level indicates, a property known as \emph{overcoverage}. Procedures that explicitly and intentionally over-cover are common in physics \cite{cls,pcl} and produce conservative claims about nature.

Constructing intervals by a \emph{Neyman construction} \cite{neyman_construction}, as illustrated in figure \ref{fig:neyman}, guarantees coverage. We assign each possible $(s, T)$ to combination to an \emph{Accept} or \emph{Reject} region, ensuring $T(s)$ has $\geq 90\%$ probability to fall in the \emph{Accept} region at every $s$. After observing $T(s)$, our interval consists of the $s$ values for which $T(s)$ is in the `Accept' region. If this set of values is not an interval, we would conservatively report the shortest interval that includes all of the set.

For 90\% confidence level upper limits, we should reject the 10\% most deficit-like observations for each $s$, as shown in figure \ref{fig:neyman}. 
For two-sided intervals, we instead reject some extreme excesses as well as some deficits; there are many ways of balancing these two \cite{feldman_cousins, cls, pcl, knut_fc}.
Regardless, we need the \emph{quantile function} $Q_q(s)$ (inverse cumulative distribution function) of the test statistic at each $s$, with $q \in [0,1]$; sketched as thin black lines in figure \ref{fig:neyman}. For 90\% confidence upper limits, we need the tenth percentile $Q_{0.1}(s)$; more generally, for $\mathrm{CL}$ confidence level limits we need $Q_\alpha(s)$ with $\alpha = 1 - \mathrm{CL}$. The threshold $Q_\alpha(s)$ can be estimated by computing $T(s)$ on simulated datasets with different $s$. Asymptotic approximations for $Q_\alpha$ are sometimes available \cite{asymptotic, nature_review}.

Neyman constructions with discrete statistics such as $N$ tend to give conservative results, i.e.~weaker upper limits. This is because the threshold $Q_\alpha(s)$ will always lie in a nonzero probability mass, which has to be assigned to the `Accept' region entirely. To mitigate this, $N$ can be made continuous artificially by adding a random number between 0 and 1 to $N$ \cite{augmented}. We will use this to fairly compare $N$ to continuous statistics in two comparisons below, in figures \ref{fig:regioncomp_plaw} and \ref{fig:statscomp_plaw}. Deficit hawks do not use this. %

\subsection{Deficit Hawks}
\label{sec:deficit_hawks}

\begin{figure*}
    \centering
    \includegraphics[width=\linewidth]{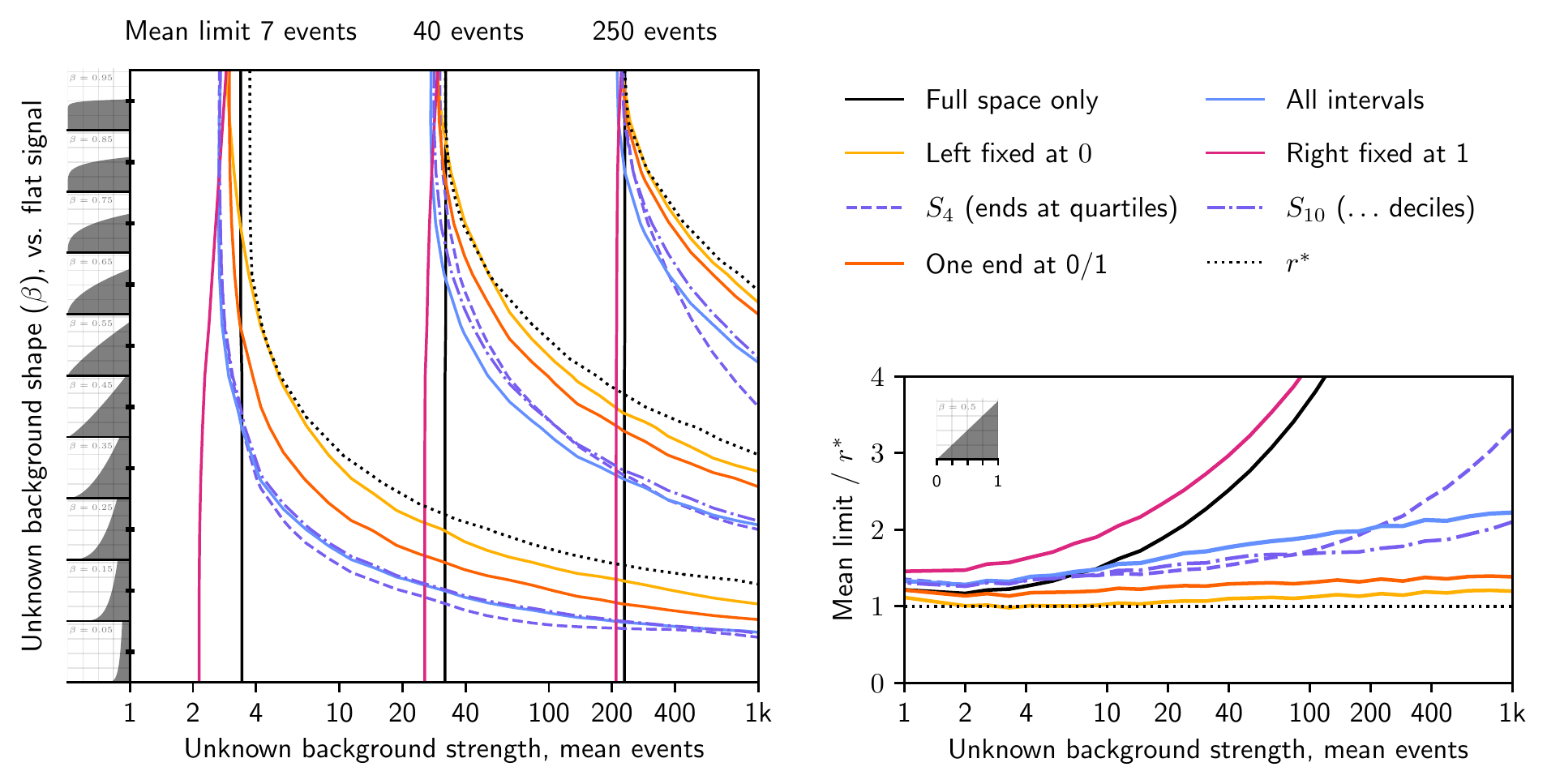}
    \caption{
    Performance of deficit hawks that test different region sets, for different unknown background strengths (horizontal axes), and for the left panel, different shapes (vertical axis, sketched in coordinates where the signal is flat).
    \emph{Left}: Successive contours connect, from left to right, scenarios at which the mean 90\% confidence level upper limit on the signal strength is 7, 40, and 250 events, respectively.
    \emph{Right}: Mean upper limits for different strengths of the triangular unknown background (halfway up the left panel), relative to testing only $r^*$.
    Colors indicate different region sets: blue for all intervals, gold for $(0, \ldots)$ and red for $(\ldots, 1)$ intervals, orange for the union of the previous two, and dashed/dash-dotted purple for $S_4$/$S_{10}$, respectively (i.e.~intervals ending at quartiles/deciles). The solid black line is for a Poisson test in the full space; the dotted black line is for a Poisson test with an added uniform(0,1) random number in $r^*$, the (unknown) optimal region for each particular unknown background. Small wiggles in some lines are due to finite Monte Carlo statistics.
    }
    \label{fig:regioncomp_plaw}
\end{figure*}

A \emph{deficit hawk} statistic $H$ is defined as
\begin{equation}
    \label{eq:deficit_hawks}
    H[T,R] = \min_{r\in R} T_r, 
\end{equation}
with $R = \{r \subseteq X\}$ a (possibly infinite) set of (generally overlapping) regions or cuts of the measured space, and $T_r$ is a test statistic $T$ computed using only the region $r$.
In the example of figure \ref{fig:simple_example}, $R$ is the set of intervals $x\in (0, \ldots)$, representing all central cylindrical fiducial volumes.
For most of this work, we choose $T = t$, the signed likelihood ratio (eq.~\ref{eq:signed_plr}), though section \ref{sec:stat_choice} will consider alternative statistics.

Informally, deficit hawks `choose' a region in which $T$ is low, i.e.~which appears to show a strong deficit. If an unknown background is strong, and appears less in one region than others (compared to the expected signal), this less-affected region will show the strongest deficit, and will be chosen by $H$.
Thus, limits set using $H$ remain strong if the unknown background is sufficiently distinguishable from the signal. We will see this quantitatively in Monte Carlo performance tests below.

This robustness is not free. If there is little or no unknown background, $H$ may pick a region showing a chance deficit, i.e.~an underfluctuation of the known background or true signal, if either of these are present. The more regions we test, the more often chance deficits are selected, and the lower (more extreme) the threshold $Q_{\alpha}(s)$ below which we can exclude $s$.
Thus, for best results, we should test a small set of regions, in which at least one will likely have comparatively little unknown background. Section \ref{sec:region_choice} will examine the effects of different region choices in detail.

Finally, we must ensure that deficit hawks give \emph{correct} upper limits even if unknown backgrounds appear. Because we set limits with a Neyman construction, we have guaranteed coverage for the case of \emph{no} unknown background. If an unknown background appears, it can only add events to the data. Thus, as long as $T$ always \emph{increases} (or stays constant) when events are added, unknown backgrounds can only weaken upper limits. For simple statistics such as $t_0$, this is easy to check by differentiation. Thus, unknown backgrounds do not destroy coverage; they only increase the overcoverage -- that is, deficit hawks are conservative.

Appendix \ref{sec:proofs} proves that signed likelihoods $t(s)$ only increase (or stay constant) when events are added, provided that higher $s$ values increase (or leave constant) the expected events in every subset of the space of measured properties $X$. This is clearly true for signal strength parameters such as cross-sections. Thus, deficit hawks can use intricate likelihoods, including known backgrounds (if modeled correctly), external constraints, and multiple discrimination dimensions. Nuisance parameters are generally not allowed, though see section \ref{sec:detections} and \ref{sec:future_research}.

\subsection{Choosing regions}
\label{sec:region_choice}

Deficit hawks should test regions that may each be optimal (or at least suitable) for different possible unknown backgrounds.
If we know nothing of the background, any region could be optimal -- but we should still exercise restraint in picking the region set $R$ to test, as testing too many regions lowers the critical value $Q_\alpha$. 

This is an opportunity to leverage domain knowledge: we often know, or at least suspect, that some backgrounds are more likely than others.
For instance, in the example of figure \ref{fig:simple_example}, we tested different maximum radii, but not different \emph{minimum} radii, or unions of disconnected intervals.
It is \emph{safe} to use background speculations to pick a set of regions to test: if the background is different than we thought, e.g.~it somehow \emph{does} appear mostly in the center of the detector, limits will be weaker, but still correct. This is no different, and no more subjective, than choosing a single fiducial volume or data quality cut, which experiments routinely do -- it is merely more flexible.

Figure 3 compares the performance of deficit hawks that test different region sets $R$, for various unknown background sizes and shapes. As in figure \ref{fig:simple_example}, the signal is flat in $x\in [0,1]$, and the deficit hawks use $T=t=t_0$. The right panel tests different strengths of one background shape (triangular in $x$), while the left panel shows a grid scan over background sizes and shapes -- more precisely, power laws $\mathrm{PDF}(x) \sim x^{-1 + 1/\beta}$, with $\beta \in (0,1]$ plotted along the vertical axis. Contours connect scenarios that give the same mean upper limit, so methods whose curves lie further to the right tolerate stronger backgrounds. The background is indistinguishable from the signal at $\beta=1$ (top), triangular at $\beta=0.5$ (middle), and approaches a delta function for $\beta\rightarrow0$ (bottom). Figure \ref{fig:simple_example} used $\beta = 0.25$ with 20 mean unknown background events.

The dotted lines in figure \ref{fig:regioncomp_plaw} show the performance of testing the optimal single region $r^*$ for each unknown background, using a Poisson test with an added random number. In figure \ref{fig:simple_example}, $r^*$ would be at the minimum of the dashed black curve. Since the background is unknown, $r^*$ is unknown -- but it is still a useful ceiling on the performance of deficit hawks \footnote{It is unclear if testing $r^*$ with a random-augmented Poisson test is better than any other deficit hawk. Deficit hawks that test multiple regions may essentially do the `smoothing' of the Poisson test with data outside $r^*$, which has more information than a purely random number. This could explain why the `Left fixed at 0' deficit hawk in figure \ref{fig:regioncomp_plaw} seems to slightly outperform the $r^*$ test for some unknown backgrounds. However, the effects are small enough that we cannot discount numerical errors. %
}.

As in figure \ref{fig:simple_example}, we see that testing intervals of the form $x \in [0, \ldots)$, i.e. trying different maximum values, usually performs about as good as testing only the unknown ideal region $r^*$.
Thus, even if we roughly knew the strength and shape of the unknown background, there would be relatively little benefit in testing an even smaller set of regions.
Unless the background is small or very similar to the signal, testing \emph{all} intervals instead is clearly worse, and testing only the full space is worse still.

\begin{figure}
    \centering
    \includegraphics[width=\columnwidth]{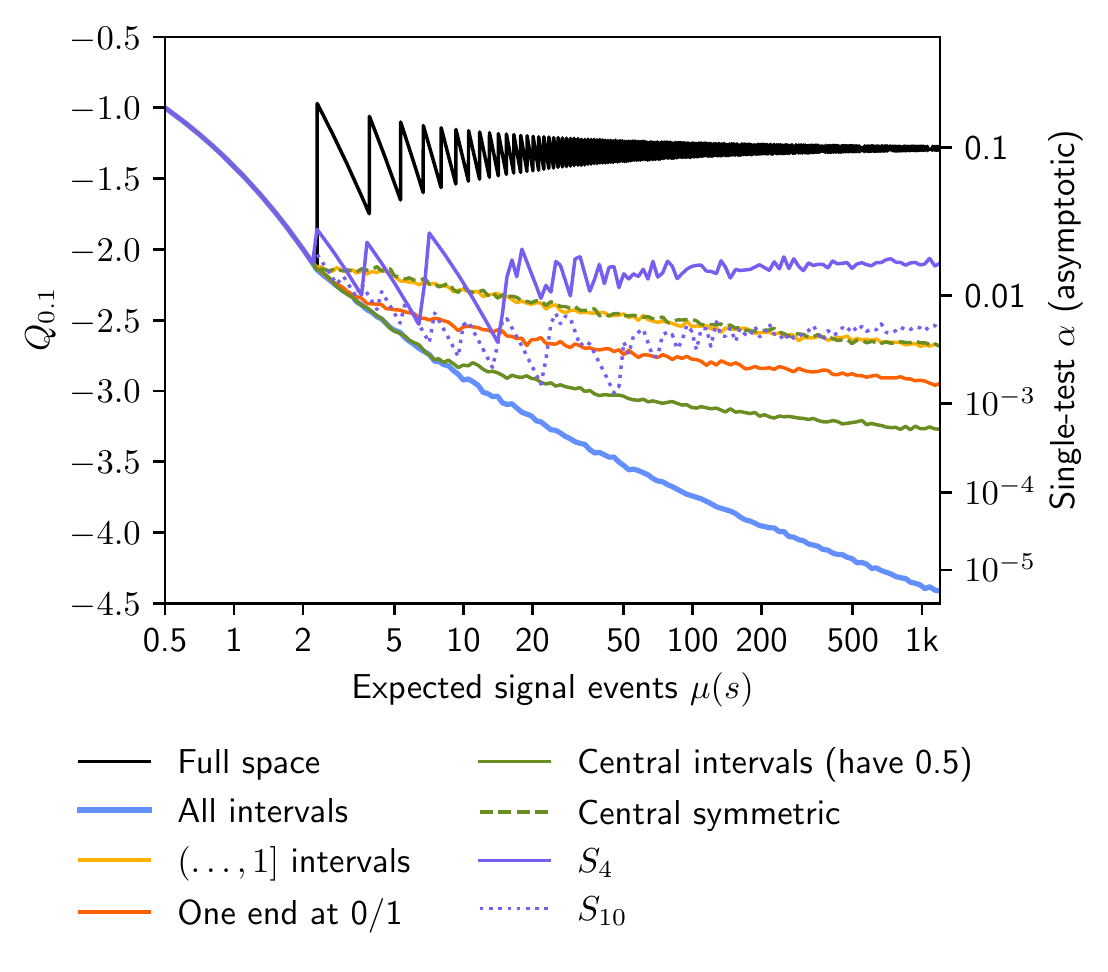}
    \caption{Threshold values $Q_{0.1}$ for deficit hawks with different region choices, including those shown in figure \ref{fig:regioncomp_plaw} and appendix \ref{sec:extra_scen}, for an experiment without known background using a deficit hawk with $T = t = t_0$. The right axis shows the equivalent $\alpha$ for a single asymptotic ($\mu(s) \rightarrow \infty$) test. Quick jumps in the black and purple lines are due to the discreteness of the test statistic; tiny wiggles in the lines are due to finite Monte Carlo statistics and the finite grid in $\mu(s)$. The black curve is calculated exactly.
    }
    \label{fig:q_alpha}
\end{figure}

Most region choices perform similarly when there is only a true signal.
That may seem surprising, given that testing many regions significantly lowers $Q_\alpha$, as shown in figure \ref{fig:q_alpha}. However, if we test more regions, each dataset also has more opportunities to cross $Q_\alpha$. Only when a sufficiently large and localized unknown background makes many of these tests futile does the low $Q_\alpha$ start to bite. This is why testing all intervals performs worse in the lower half of figure \ref{fig:regioncomp_plaw}. There is no single measure of the `penalty' to sensitivity for testing more regions: that, too, depends on the unknown background.

Testing sets of regions such as `all intervals' may seem computationally expensive, but in practice, we need only test intervals ending at observed events, 0 or 1. As mentioned in section \ref{sec:simple_example}, an interval with empty space beside it will give a lower $t$ than the same interval with the empty space included. Appendix \ref{sec:proofs} and \ref{sec:compute} show a proof of this property.

Even testing only regions that border on events might be too expensive if many events are observed. Alternatively, we could test the set $S_k$ of all intervals with endpoints at the $k$-quantiles of the signal model, as proposed in \cite{optitv2}. For example, in coordinates where the signal is flat on $[0,1]$, $S_4$ consist of the ten intervals $[0, 0.25]$, $[0.25, 0.5]$, $[0.5, 0.75]$, $[0.75, 1]$, $[0, 0.5]$, $[0.25, 0.75]$, $[0.5, 1]$, $[0, 0.75]$, $[0.25, 1]$ and $[0, 1]$ ending on quartiles of the signal distribution.
Figure \ref{fig:regioncomp_plaw} shows the performance of testing $S_4$ and $S_{10}$.
As $k \rightarrow \infty$, testing $S_k$ approximates testing all intervals, and does so more quickly at low $s$.

Figure \ref{fig:regioncomp_plaw} only tested backgrounds that grow (or equivalently decay) monotonically, e.g.~towards a detector edge or over time, which was rational for the example in figure \ref{fig:simple_example}.
In other cases, e.g.~searching for a (peaked) signal on top of a (flatter) background in a dimension like energy, unknown backgrounds may appear on both sides of the signal. In that case, we could e.g.~test all intervals that include some central point, moving symmetrically or asymmetrically outwards. Appendix \ref{sec:extra_scen} compares the performance of different region choices in this and several other scenarios.

\subsection{Choosing a statistic}
\label{sec:stat_choice}
This subsection examines alternatives to using the signed likelihood ratio $t$ for the per-region statistic $T$ in deficit hawks.

The statistic $T$ in a deficit hawk ranks the regions, and the lowest-scoring region `wins'. The exact values $T$ produces are irrelevant -- if two statistics agree on how to rank all possible observations, they yield the same upper limits.

Any reasonable statistic used in a deficit hawk should prefer large regions with few observed events. Formally, we might call a statistic $T$ \emph{countike} if, for any dataset $D$ and region $r \subseteq X$ in which $T_r$ is computed,
\begin{enumerate}
    \item $T_r$ \emph{increases} (or stays constant) if an extra event is added to $D$, and
    \item $T_r$ \emph{decreases} (or stays constant) when $r$ is expanded by some space without observed events.
\end{enumerate}
Section \ref{sec:deficit_hawks} discussed how the first condition guarantees coverage even if unknown backgrounds appear. The second condition enables the optimization discussed in section \ref{sec:region_choice} of testing only regions bounded by events (when testing a continuum of regions).
The count $N$ of observed events is the archetypical countlike statistic. 
To show $t_0$ is countlike, one can check $\partial t_0 / \partial N \geq 0$ and $\partial t_0 / \partial \mu \leq 0$.
Appendix \ref{sec:proofs} proves that the general signed likelihood ratio $t(s)$ is countlike. As a counterexample, the unsigned likelihood ratio $u(s)$ is not countlike: since it takes increasingly high values for both stronger deficits and stronger excesses, adding events to $D$ can lower $u(s)$ for some $s$.

Some countlike statistics may prefer large regions, even if they have a bit more events, while others prefer few observed events, even if it requires choosing a smaller region. We might call these \emph{soft} and \emph{hard} statistics, respectively.
Very soft statistics essentially only probe the full space, and very hard statistics probe only empty regions (`gaps'). 
If only gaps are tested (and there are no known backgrounds), we can use $T = -\mu$ instead of more complicated statistics, as it does not matter how the statistic scores events.

\begin{figure}
    \centering
    \includegraphics[width=\columnwidth]{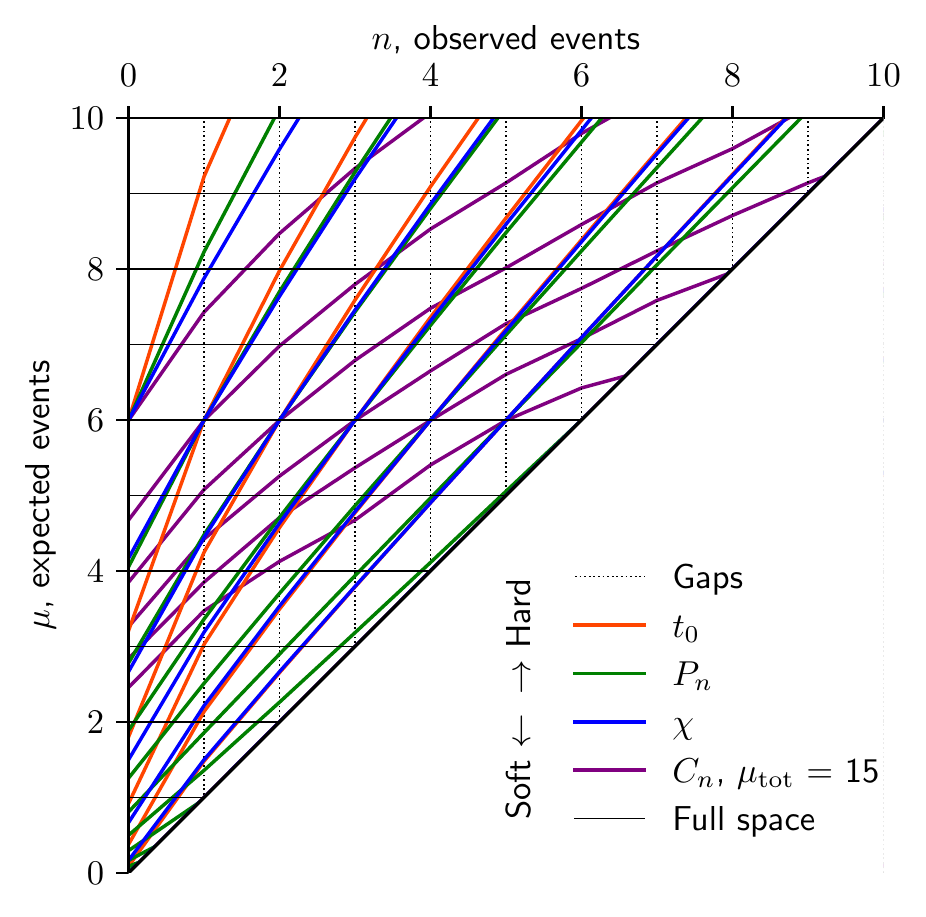}
    \caption{Indifference curves for different statistics, versus the count of observed events $n$ and expected $\mu$ events in a tested region. Along any one curve, the statistic takes the same value. Curves intersecting $(\mu = 6, n)$ with different $n$ are drawn. $C_n$ also depends on the \emph{total} expected events $\mu_\text{tot}$; we show the indifference curves for $\mu_\text{tot} = 15$ events.}
    \label{fig:indiff}
\end{figure}

\begin{figure}[t]
    \centering
    \includegraphics[width=\columnwidth]{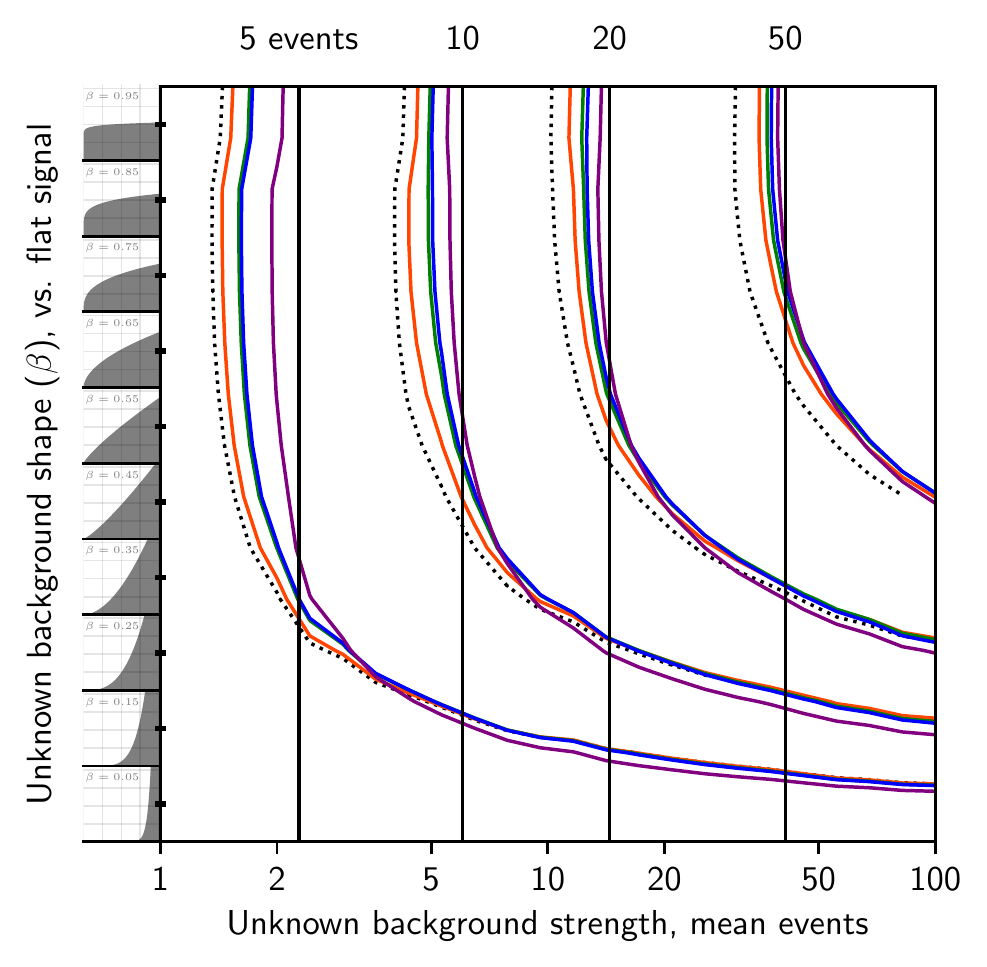}
    \caption{As figure \ref{fig:regioncomp_plaw}, but comparing the different statistics shown in figure \ref{fig:indiff}, with the same color coding.
    Successive contours connect, from left to right, scenarios at which the mean upper limit on the signal strength is 5, 10, 20, and 50 events, respectively. 
    The colored lines test all intervals, the solid black line is for counting events in the full space (with a uniform(0,1) random number added to the count), and the dotted line is for testing empty intervals / gaps (with $T=-\mu$, or equivalently any countlike statistic that strictly decreases with $\mu$). Small-scale variations reflect the finite $(16 \times 21)$ grid of tested models and the finite number (1500) of Monte Carlo trials used to estimate the mean limit at each point.
    }
    \label{fig:statscomp_plaw}
\end{figure}

Three statistics in particular have been used in earlier literature. 
The `$p_\text{max}$' method introduced in the appendix of \cite{optitv} is a deficit hawk that tests all intervals with the Poisson CDF of $N$:
\begin{equation}
    \label{eq:pn}
    P_n(s) = P(N \leq n | \mu(s)),
\end{equation}
with $n$ the observed value of $N$ and $\mu$ the expected number of events. Using
\begin{equation}
\chi = (N-\mu)/\sqrt{\mu} ;
\end{equation}
was suggested in \cite{optitv2}. Finally, Yellin's optimum interval method is a deficit hawk that tests all intervals with $C_n$, a statistic detailed in \cite{optitv}. $C_n$ depends on the expected events in the full space $\mu_\text{tot}$, besides the $\mu$ and $n$ of the tested region. %

Figure \ref{fig:indiff} visualizes how these statistics rank observations, using `indifference curves' that connect  $(n, \mu)$ combinations on which a statistic takes the same value. Hard statistics have steep curves, since they require much higher $\mu$ to offset an increase in $n$. 
Testing only gaps correspond to infinitely steep indifference curves: any $n>0$ will disqualify a region.
Testing only the full space corresponds to horizontal curves: the region with highest $\mu$ is always preferred.
Clearly, $P_n$ and $\chi$ are slightly softer but overall quite similar to $t_0$. $C_n$ is much softer at the $\mu_\text{tot}$ shown; at higher $\mu_\text{tot}$ it hardens and approaches $P_n$.

Figure \ref{fig:statscomp_plaw} compares the sensitivity of deficit hawks that use the statistics discussed above, similar to figure \ref{fig:regioncomp_plaw}. Comparing figures \ref{fig:statscomp_plaw} and \ref{fig:regioncomp_plaw}, note that figure \ref{fig:statscomp_plaw} has a smaller horizontal scale, and that its contours are drawn at more closely spaced levels. Clearly, which region set $R$ to test is far more important than which statistic $T$ to use for comparing the regions.

The small differences we do see in figure \ref{fig:statscomp_plaw} are consistent with how soft or hard the different statistics are. In the upper left, the unknown component is weak and signal-like. Thus, large regions give better results, soft statistics such as $C_n$ do well, and testing only the full space is best. In the lower right, the situation is reversed: backgrounds are large and clearly different from the signal, so harder statistics like $t$ do better.

There is no `optimal' statistic $T$. Because reasonable statistics perform similarly, choosing a particular statistic on purpose for a problem is difficult to motivate.
This work uses $t$ because many experiments already use likelihood ratios, since these have distinct advantages when \emph{known} backgrounds are present (as discussed below).
With a fully unknown background, other statistics are perfectly good choices.
The deficit hawk shown in figure \ref{fig:simple_example} outperforms Yellin's optimum interval method because it tests fewer regions -- the choice of statistic has little to do with it.

\section{Extensions}

\subsection{Known backgrounds}
\label{sec:known_bg}

\begin{figure}[t]
    \centering
    \includegraphics[width= \linewidth]{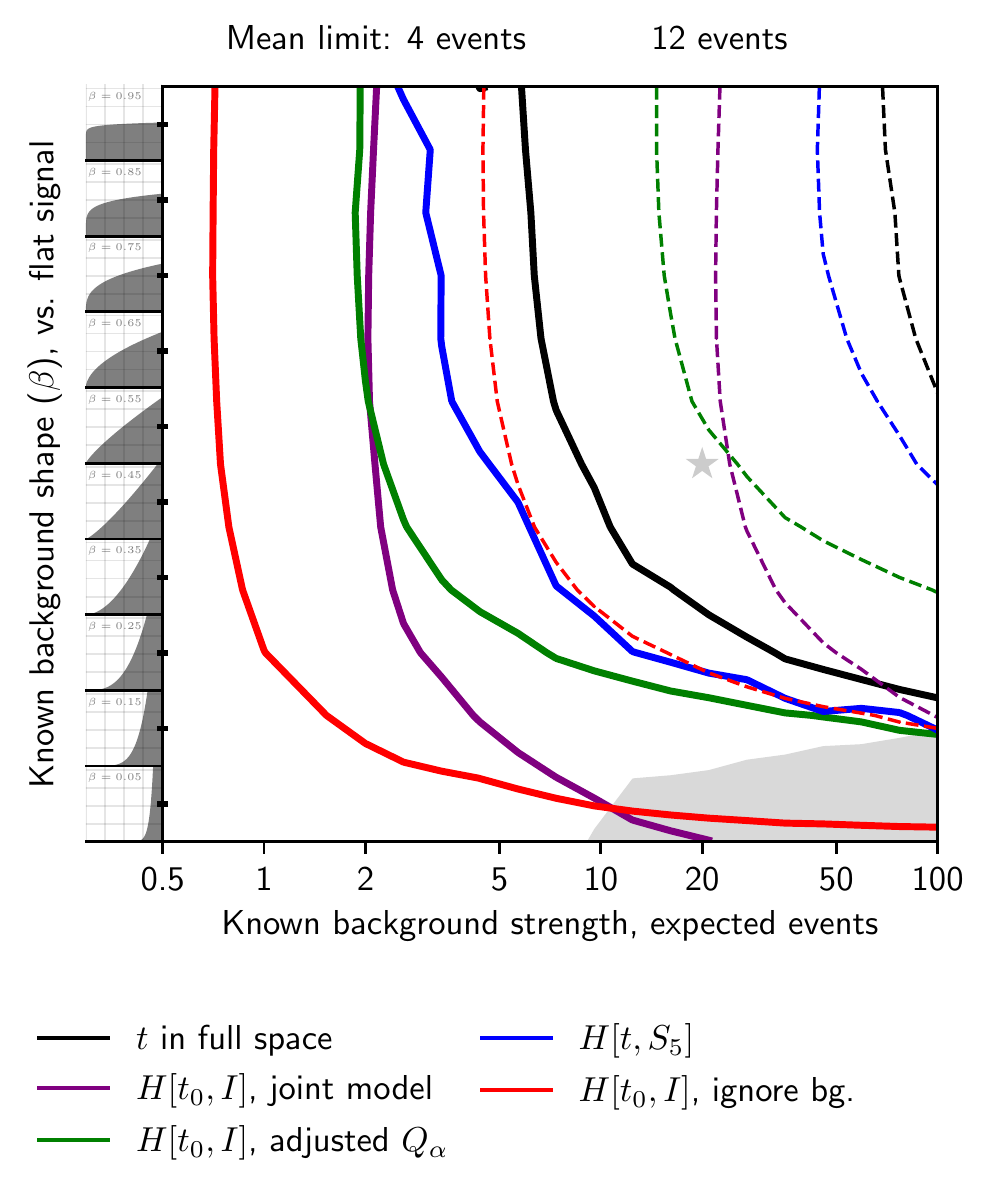}
    \caption{
    As figures \ref{fig:regioncomp_plaw} and \ref{fig:statscomp_plaw}, for experiments with only \emph{known} backgrounds of different strengths and shapes. Solid and dashed lines connect scenarios in which the mean upper limit on the signal strength is 4 and 12 events, respectively.
    Black lines are for for testing the full space with the signed likelihood ratio $t$, and blue lines are for using the full likelihood $t$ as the test statistic and $S_5$ as the region set. The other lines are for using $t_0$ and testing all intervals ($I$): green for the adjusted threshold method, purple for the joint model method, and red for ignoring the known background completely. 
    The gray shaded area marks scenarios in which ignoring the known background outperforms the joint model method. 
    The blue contours are rougher than others because they were estimated with lower Monte Carlo statistics, both for the toy limits and the $Q_{0.1}$ estimate.
    The gray star is used in figure \ref{fig:statscomp_mixed}.
    }
    \label{fig:statscomp_known}
\end{figure}

\begin{figure}[t]
    \centering
    \vspace{0.53cm}
    \includegraphics[width= \linewidth]{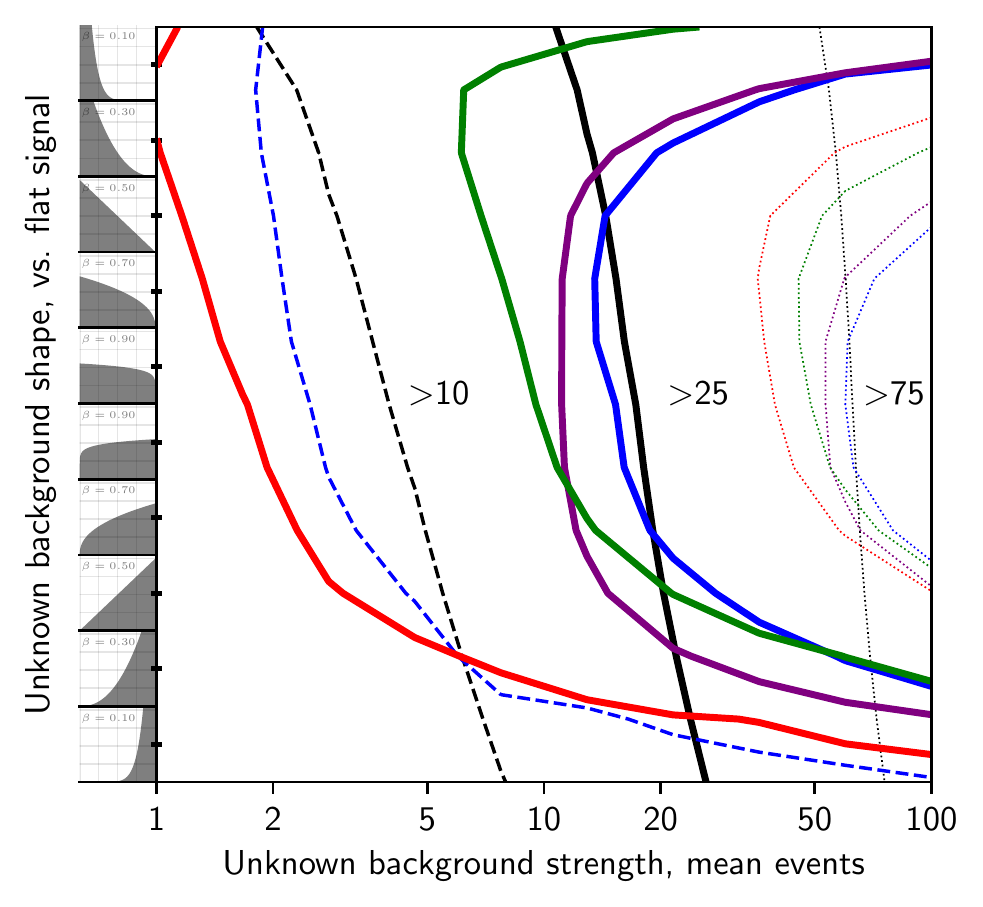}
    \caption{
    As figures \ref{fig:regioncomp_plaw},  \ref{fig:statscomp_plaw} and \ref{fig:statscomp_known}, for an experiment with both known and unknown backgrounds. The known background is held fixed at the gray star drawn in figure \ref{fig:statscomp_known}; the unknown background is as described by this figure's axes. Colors are as in figure \ref{fig:statscomp_known}; dashed, solid, and dotted lines connect scenarios in which the mean upper limit on the signal strength is 10, 25, and 50 events, respectively. The bottom half tests the same background shapes as earlier figures (see section \ref{sec:region_choice}); in the top half, the unknown background is mirrored ($x \rightarrow 1 - x$). Thus, the unknown component now has the same shape as the signal in the \emph{center} of the figure, rather than at the top.
    As before, do not confuse the red contour for 25 events with the contours for 10 events. Only the black and blue lines have a dashed contour; for the other methods, mean limits are above 10 everywhere.
    }
    \label{fig:statscomp_mixed}
\end{figure}

Analysts can often do more than just speculate about backgrounds: they may propose a partial background model. This means they assume backgrounds are \emph{at least} as large as a model predicts.

With a partial background model, the signed likelihood ratio $t$ becomes the obvious choice for the per-region statistic $T$, as $t$ discriminates the signal from known backgrounds at a per-event level. 
Note the likelihood can use any observable we have models for, whether or not it is used in the deficit hawk's region choices.
For example, even if a deficit hawk tests only several energy ranges, the likelihood can use reconstructed position or other observables for discrimination, instead of or in addition to energy.

Figure \ref{fig:statscomp_known} compares the performance of a deficit hawk using $t$ (in blue) with one that uses $t_0$ and ignores the background model (in red), if \emph{all background is known}, in the one-dimensional example explored in figures \ref{fig:regioncomp_plaw} and \ref{fig:statscomp_plaw}. Because $t$ is expensive to compute, we tested only the 25 $S_5$ regions with it; figure \ref{fig:regioncomp_plaw} indicates this should perform similar to testing all intervals for signals with modest predicted event count.

As expected, a deficit hawk using $t$ performs much better than a deficit hawk that treats all background as unknown -- by a factor $\sim \! 3$ near the center of the figure. If truly \emph{all} background is known, there is of course an even better method: a simple likelihood ratio test in the full space (black in figure \ref{fig:statscomp_known}). Deficit hawks prove their worth when unknown backgrounds appear.

Figure \ref{fig:statscomp_mixed} repeats our comparison for scenarios with known \emph{and} unknown background. Specifically, we chose one fixed \emph{known} background, marked by a star in figure \ref{fig:statscomp_known} (20 expected events, and a linearly rising distribution). Figure \ref{fig:statscomp_mixed} then studies many possible \emph{unknown} backgrounds that appear on top of this.
Indeed, we see that a deficit hawk using $t$ can significantly outperform a likelihood test in the full space. %
Testing the full space is preferable only if the unknown background is weak and similar to the signal -- and even then, the deficit hawk does not perform much worse.

Although deficit hawks that use $t$ perform well, their computational cost may be prohibitive, e.g.~for analyses that test many regions, or simply unnecessary, e.g.~because only a tiny part of the background is known. In such cases, we can fall back to using $t_0$ or other simple countlike statistics, but still have two options to leverage our partial background knowledge:
\begin{itemize}
    \item \emph{Joint model}: take the \emph{total} expected events in each region for $\mu$ in the hawk's statistic.
    \item \emph{Adjusted threshold}: take the expected \emph{signal} events for $\mu$ as usual. However, use the known background when computing the threshold $Q_\alpha(s)$.
\end{itemize}

The joint model method was proposed in \cite{optitv}, and makes the deficit hawk search for deficits with respect to the summed signal and background model. 
The adjusted threshold method instead computes the deficit hawk just as if all background were unknown. However, since we know there is some background, we expect to see higher (more excess-like) values than in a background-free experiment at any $s$. Thus, our threshold $Q_\alpha(s)$ is higher, and we can exclude an $s$ sooner, i.e.~at less extreme $H(s)$ values, than if we knew nothing of the background.

Figures \ref{fig:statscomp_known} and \ref{fig:statscomp_mixed} show the performance of these methods in green and purple. 
The joint model method can perform worse than treating all background as unknown (lower right of figure \ref{fig:statscomp_known}), because this method can favor regions where little (or in principle, no) signal is expected, or where the known background has a significant statistical underfluctuation. This becomes problematic if the known background is large and different from the signal.

The adjusted threshold method is guaranteed to do better than ignoring the known background: it computes the same statistic as in the fully-unknown case, then compares it against a higher threshold $Q_\alpha$. Its main weakness is that the statistic does not distinguish known and unknown backgrounds. Thus, if two regions have the same expected signal, the region with fewest observed events will win, even if the other region has a high expected background. This hurts the performance when the unknown and known background have opposite shapes, as in the upper half of figure \ref{fig:statscomp_mixed}.

In summary, if there are substantial known backgrounds, a deficit hawk with $t$ should give good results, even if it tests a limited number of regions. If only a tiny fraction of the background is known, using $t_0$ with the adjusted critical value method can be a good alternative. If, in the latter situation, the known and unknown background are likely to have very different shapes, the joint model method will work better instead.

\subsection{Underfluctuations and Sensitivity}
\label{sec:underflucts}

Background underfluctuations can force excessively strong exclusions from any analysis that subtracts known backgrounds. Significant \emph{unknown} backgrounds make such  underfluctuations unlikely, but the possibility should still be considered. A simple and common remedy is to cap exclusion limits at the 15.9th percentile (``$-1\sigma$'') of expected upper limits. For two-sided analyses, this is an approximation to power-constrained limits \cite{pcl}.

Analyses often show a `Brazil' band of expected upper limits along with their results \cite{atlas,cms,dd_conventions}. For experiments that know most, but not all, of their background, showing this band remains valuable -- if accompanied with a prominent caution that limits well above the band may merely indicate significant unknown backgrounds.

Analyses with mostly or completely unknown backgrounds can instead show the mean or median limit expected without unknown backgrounds, rather than a Brazil band. This is not an expected performance, but still a useful reference.

\subsection{Detection claims}
\label{sec:detections}

\begin{figure*}
    \centering
    \subfloat{{\includegraphics[height=6.46cm]{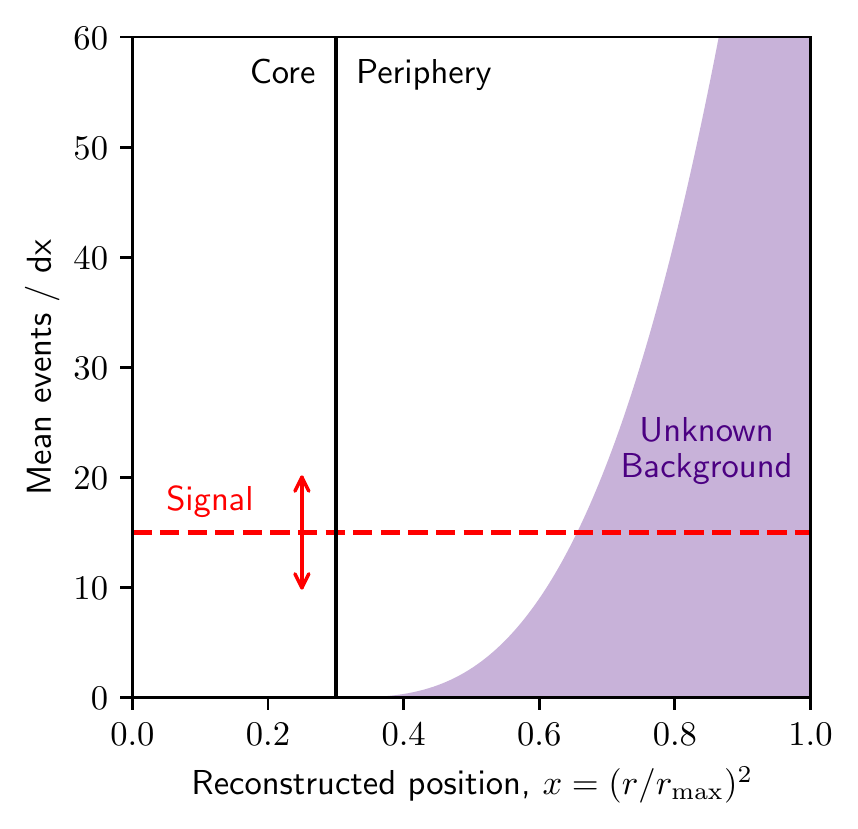}}}%
    \qquad
    \subfloat{{\includegraphics[height=7.75cm]{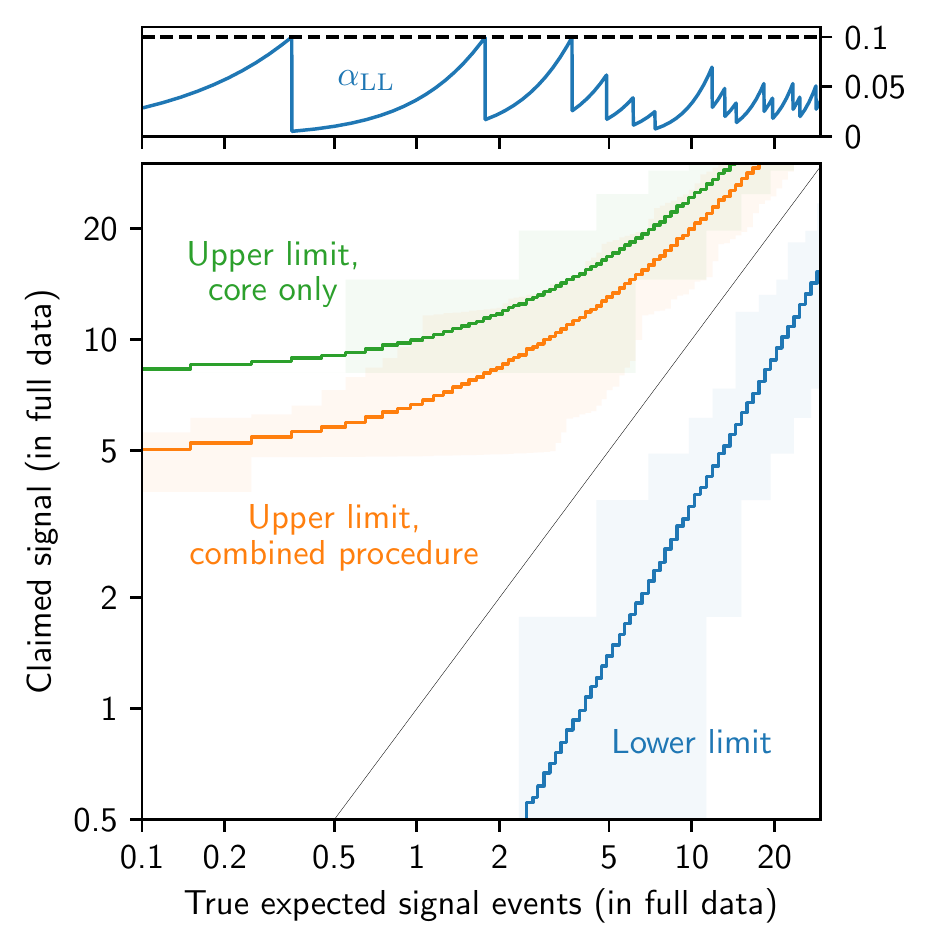} }}%
    \caption{An example of combining deficit hawks with a discovery-capable analysis. The experiment is as in figure \ref{fig:simple_example}, except that the inner $30\%$ core of the experiment is known to be background-free. The periphery has an unknown background, here taken to be the same as in figure \ref{fig:simple_example} (squeezed to the smaller range), as illustrated on the left plot.
    The bottom right plot shows the mean upper (green) and lower (blue) limit from a standard Feldman-Cousins \cite{feldman_cousins} counting analysis in the core. The orange line shows the mean upper limit from the procedure in section \ref{sec:detections}, which gives the same lower limits as the core Feldman-Cousins analysis. The top right panel shows $\alpha_\text{LL}$, the probability of that the lower limit excludes the true signal. Thin bands shade the region between $\pm 1 \sigma$ quantiles of results, if this exist. Because the core analysis can only produce a discrete set of results, these bands are highly discontinuous.
    }%
    \label{fig:discovery_study}%
\end{figure*}

If arbitrary unknown backgrounds can appear, discovery claims are impossible: the unknown background might look exactly like a new physics signal. While most experiments have unknown backgrounds, many can select a \emph{core} region of data in which they are absent or insignificant.
That core region can be employed to set two-sided confidence intervals -- but this leaves the remaining \emph{peripheral} data unused.

Instead, we can do a two-step Neyman construction similar to \cite{knut_fc} to naturally integrate deficit hawks into discovery analyses. Here, as before, $\alpha = 1 - \mathrm{CL}$, with $\mathrm{CL}$ the desired confidence level.
\begin{enumerate}
    \item Use any desired two-sided procedure in the core region, e.g.~Feldman-Cousins \cite{feldman_cousins} or profile likelihoods with nuisance parameters. Keep only the \emph{lower limit}, i.e.~the detection claim, if there is one. 
    \item For each $s$, determine (perhaps with Monte Carlo) the probability $\alpha_\text{LL}(s) \leq \alpha$ of mistaken exclusions of the true signal by the lower limit. 
    \item Determine an upper limit with a deficit hawk, using an $s$-dependent confidence level $\alpha_\text{UL}(s) = \alpha - \alpha_\text{LL}(s)$. That is, use $Q_{{\alpha_\text{UL}}(s)}(s)$ rather than $Q_\alpha(s)$ as the threshold value. The deficit hawk should use a likelihood ratio $t$ to test the core region and one or more extensions of it. If the core analysis has nuisance parameters, fix them at conservative values for limit setting (i.e.~assume a low background model).
\end{enumerate}

This procedure has (over)coverage by construction. For any $s$, the chance of falsely excluding it is $\alpha_\text{LL}(s)$ for the lower limit (step 1) and $\leq \alpha - \alpha_\text{LL}$ from the upper limit (step 3), even if unknown backgrounds appear in the periphery. Thus the total false exclusion probability is at most $\alpha$ for each $s$, as required. 
While any method could be used in step 3, likelihood-wielding deficit hawks make conceptual sense -- they do the same analysis as in the core, with different cuts and frozen nuisance parameters.

As an example, consider figure \ref{fig:discovery_study} (left), which sketches a situation similar to figure \ref{fig:simple_example} in which the core $30\%$ of the detector ($x < 0.3$) is known to be background-free. Known backgrounds can be handled straightforwardly (see section \ref{sec:known_bg}), but are omitted in this example for simplicity. The peripheral $70\%$ ($x > 0.3$) of the detector has an unknown background -- in this case, the background shown in figure \ref{fig:simple_example} squeezed into $x \in [0.3, 1]$. As step 1, we set simple $90\%$ Feldman-Cousins \cite{feldman_cousins} intervals using the count of events in the core region. Step 2 computes $\alpha_\text{LL}$, the probability of false discoveries, shown in the top right of figure \ref{fig:discovery_study}. It approaches $\alpha/2 = 0.05$ for large signals, and has sawtooth-like behavior at small scales because event counts are discrete. Finally, in step 3, we use a deficit hawk that test the core cylinder $x \in [0, 0.3]$ and all \emph{larger} central cylinders. The deficit hawk uses a more stringent critical threshold $Q_{0.1 - \alpha_\text{LL}(s)}(s)$ than the usual $Q_{0.1}(s)$. Figure \ref{fig:discovery_study} (right) shows the mean upper and lower limits from this procedure. The upper limit can be nearly twice as strong as the core Feldman-Cousins analysis would produce, while we produce the same lower limits.

This procedure has some peculiarities which, though unimportant in the example above, may require more detailed study in other situations.

First, although the procedure has (over)coverage, it can produce very narrow and even empty intervals in an exceptional circumstance: when we see many events in the core and few in the periphery.
In the example of figure \ref{fig:discovery_study}, we get an empty interval if we observe $\geq 13$ events with $x<0.3$, and no other events. Clearly this is vanishingly unlikely, even if the periphery had no unknown background. In situations where empty intervals are a real possibility -- perhaps due to large \emph{known} backgrounds that could underfluctuate -- analysts could decide to publish the lower limit / discovery claim only if some threshold discovery significance is met in the core \emph{and} all (core+periphery) extensions tested. Peripheral regions with high unknown backgrounds will easily exceed the threshold.

Second, the discarded upper limit from the core analysis could be stronger than that of the deficit hawk. Though unfortunate, this is neither problematic nor common: the deficit hawk always tests the core region too, and should at least be somewhat effective in mitigating the unknown background. In the example of figure \ref{fig:discovery_study}, the deficit hawk produced the stronger upper limit in $>99\%$ of trials without true signal. If a strong signal is present, the upper limits from the two analyses perform more similarly; at $\mu_\text{sig} = 30$ in this example, the upper limit from the core is best in $\sim 30\%$ of trials.

Finally, the deficit hawk (step 3) must always allow hypotheses for which the core analysis (step 1) saturates $\alpha$. In the example of figure \ref{fig:discovery_study}, if the truth is near near $1.1$ signal events in the core, or 3.7 in the full data, the lower limit alone would exclude it up to $10\%$ of the time \footnote{The lowest possible \emph{upper} limit from the core analysis is 2.44 signal events in the core \cite[table IV]{feldman_cousins}, or 8.2 events in the full data. The lower limit does not always saturate the full $\alpha$ below this due to the overcoverage inherent in using a discrete statistic.}. Thus, step 3 can never exclude signals below 3.7 events in the full data.

Sometimes, as in this example, this is no problem: the signal rates for which $\alpha_\text{LL} = \alpha$ would be too low to exclude by any method. However, for small cores or low unknown backgrounds, this is unacceptable: we should be able to exclude hypotheses that predict many events in the full data if the unknown background is small enough. In such cases, we should ensure the core analysis does not spend our full $\alpha$ budget. For example, we could set $92\%$ instead of $90\%$ Feldman-Cousins intervals in the core, or less arbitrarily, $95\%$ one-sided lower limits with some threshold discovery significance.

\subsection{Multiple dimensions}
\label{sec:multidim}
Above, we illustrated the deficit hawk method with experiments that measure one quantity, but experiments that measure multiple observables per event can also use deficit hawks. 

As mentioned in section \ref{sec:known_bg}, to discriminate known backgrounds, deficit hawks can benefit from using likelihoods with more or different observable dimensions than the one(s) distinguishing different tested regions. However, the regions tested in deficit hawks can themselves also span multiple dimensions. Earlier studies have already explored this for the optimum interval and maximum gap methods \cite{optitv2, max_patch, pico}.

\begin{figure}
    \centering
    \includegraphics[width=0.7  \columnwidth]{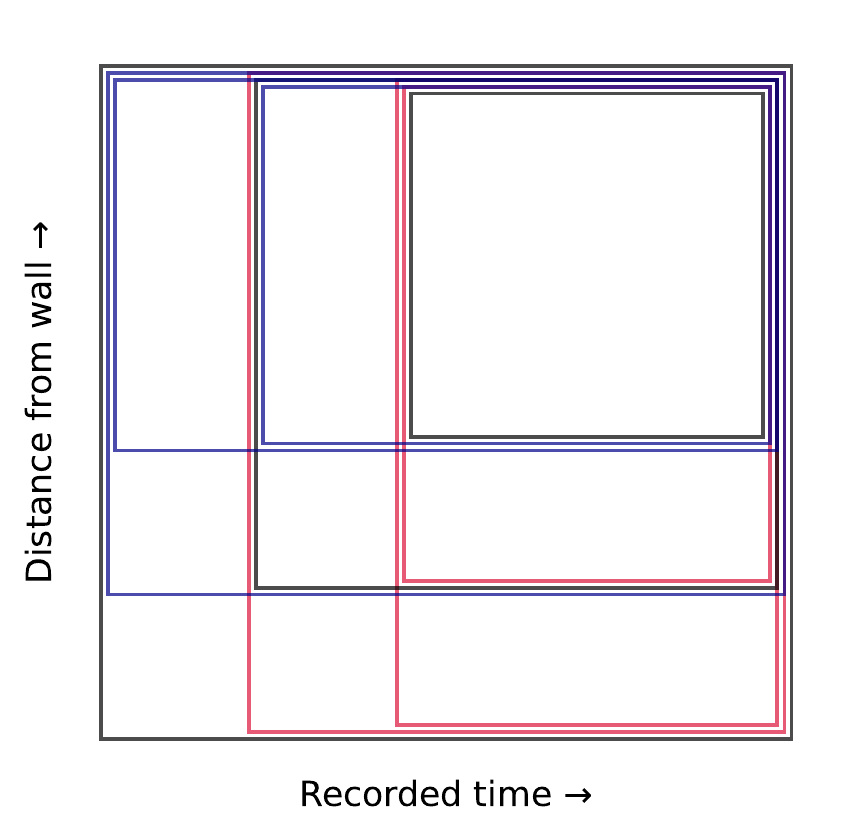}
    \caption{Nine example regions that an experiment measuring times and positions of events could test. Tiny offsets of the rectangles are for visualization.}
    \label{fig:multidim}
\end{figure}

For example, suppose an experiment records times and positions of events, and analysts are concerned about unknown backgrounds at early times (e.g.~decaying contaminants) as well as near detector edges. Analysts could then choose to test the rectangles illustrated in figure \ref{fig:multidim}: combinations of different time thresholds and fiducial volume cuts -- perhaps a conservative, nominal, and aggressive option in each dimension. The smallest box in the top right could perhaps be a `core' region for a discovery analysis.

If one of the observables is energy, analysts may be concerned about backgrounds at both high and low energies. In that case, they could test e.g.~$S_4$ or even all intervals in energy, and multiple time or position thresholds for each energy cut.
If analysts expect highly correlated backgrounds (e.g.~near the walls \emph{and} at early times), they could test some set of non-rectangular regions. And so forth -- whatever is sensible given speculations about the unknown background.
With regions that differ across multiple dimensions, the number of tested regions can multiply quickly. To avoid the corresponding decrease in physics reach, it is important to exercise restraint in choosing how many regions to test.

An `observable' used in a deficit hawk can also be the output of another statistical model. For example, analysts could use the signal/background likelihood ratio computed by an incomplete or possibly inaccurate background model, or even the score of a machine learning algorithm. This would fold multidimensional discrimination problems into one dimension, to which deficit hawks may be easier to apply.

For example, analysts might train a signal-background discriminator on simulated data (e.g.~\cite{lux_ml}); or sacrifice some science data for training such a classifier even though the science data may have some signal \cite{cwola}. If the signal simulation is accurate, the signal distribution in the model output dimension is known. However, we are rarely as confident in how the model scores all possible backgrounds, including rare instrumental backgrounds. This problem is ideally suited to deficit hawks.

\section{Discussion}
\label{sec:discussion}

\subsection{Objections and rebuttals}
This work showed how using deficit hawks -- i.e.~testing multiple options for cuts simultaneously -- allows experiments to leverage uncertain domain knowledge, and thereby improve their physics reach. It is worth re-emphasizing that this does \emph{not} introduce additional `subjectivity' or `assumptions' in the analysis. 
Deciding to set a single hard fiducial volume cut or energy threshold, is just as `objective' as deciding to try several predefined options and correctly account for this statistically. We must of course decide which options to try \emph{before} seeing the data. That makes making a good single choice difficult -- which is why deficit hawks are beneficial. 

Experiments may or may not have good physical reasons to test particular regions -- but in an upper limit study, those are not `assumptions' that would make the analysis incorrect when violated. Indeed, experiments often choose regions for reasons that have nothing to do with physics, e.g.~to make a fiducial mass a nice round number. If an experiment makes a poor region choice, it may see much unknown background, but its upper limits are still correct -- just weaker than they might have been. Incorrect results can of course happen if experiments decide to assume known background models (but in fact, the background is different), or assume discoveries are possible because unknown backgrounds are absent in a core region (but in fact, they are still present).

Another objection might be that choosing sensible regions to test is an impossible task if the background is truly unknown.
Absent situations such as figure \ref{fig:simple_example}, where there are strong a-priori reasons for fixing the minimum fiducial radius at zero, should we always fall back to testing broad sets like `all intervals'?

Of course, testing all intervals is a reasonable choice in some situations. But experiments already employ many ingenious ways to learn about unknown backgrounds -- often enough to make them comfortable with simple hard cuts. Any experiment willing to do \emph{that} should not balk at choosing a custom region set in a deficit hawk.

In fact, using deficit hawks stimulates analysts to study the unknown background. Even if background studies do not yield a formal model or a clear single optimal cut, deficit hawks allow analysts to benefit from their knowledge if it is correct, without risking erroneous results if it is not.

For example, analysts could sacrifice $\sim \! 10\%$ (or some other fraction) of the data, and see if any regions have few observed events despite a high signal expectation, or vice versa \footnote{Sacrificing some data to inform the inference strategy is known as `sample splitting' in statistics \cite{post_selection}, and requires assuming that observed events are mutually independent. This is true for Poisson processes like particle scatters and radioactive decays studied in particle physics. Events in real detectors may deviate slightly from a Poisson process, e.g. because it takes some time to record an event and return to a rest state. For rare-event searches, the Poisson process is generally still an excellent approximation.}. Perhaps the analysts see many events at very low energy, and therefore decide to test regions with different low-energy thresholds on the unseen remaining $\sim \!90\%$, while setting a fixed high-energy threshold. If the true background indeed rises monotonously towards low energies (in coordinates where the signal is flat), figure \ref{fig:regioncomp_plaw} shows this could yield about twice as strong limits than testing all intervals -- well-worth the sacrificed $10\%$. With multiple data quality dimensions, sacrificing a bit of data to see which regions are worth testing may be even more important.

\subsection{Signal uncertainties}
Throughout this work, we assumed the signal model is known. If it is estimated conservatively instead, that is fine too -- any `extra signal' will behave like an unknown background and weaken the upper limits. 

However, an \emph{overestimated} signal detection efficiency is problematic for deficit hawks, including the optimum interval method and its variants. Overestimated signals cause fake deficits even if true signals appear, and deficit hawks actively seek those deficits out. The more regions a hawk tests, and the smaller those regions are allowed to be, the more likely an unexpected local signal efficiency loss is to cause incorrect results (false exclusions of true signals). Compared to testing a single region, deficit hawks offer increased robustness to underestimated backgrounds, at the price of increased fragility to overestimated signals.

In astroparticle physics, this is often a sensible tradeoff. Experiments are usually designed to allow calibrations that accurately establish the signal detection efficiency, but quantifying all possible (instrumental) backgrounds is considerably more challenging. Though new physics models usually have unknown parameters (such as dark matter masses or couplings), new physics searches generally constrain individual options for these separately. 

We can also mitigate concerns about imperfect signal efficiencies in several ways. First, avoid testing regions in which the signal model is particularly uncertain, e.g.~near trigger thresholds. Second, avoid testing large numbers of arbitrarily small regions -- for example, require that some minimal fraction of the total expected signal must survive, or limit the tested regions to several hand-picked options. Finally, examine and ideally publish the region ultimately favored by the deficit hawk, to verify that the signal detection efficiency in it can reasonably be assumed as known.

\section{Conclusions and outlook}
\label{sec:conclusions}

\subsection{Summary}
Deficit hawks allow experiments to constrain new physics signals, even when unknown backgrounds are present in their data. The idea is simple: test multiple regions or cuts of the data (chosen in advance), take the region that gives the most deficit-like result, and set a valid upper limit with a threshold $Q_\alpha(s)$ estimated by Monte Carlo simulations.

As we saw in section \ref{sec:region_choice}, choosing which regions to test is the key way to optimize performance. Testing too few regions limits what background shapes the method is robust to, and testing too many regions unnecessarily weakens the limits by lowering the threshold $Q_\alpha$. Good region choices find a golden middle, testing only cuts that have a reasonable chance of fighting likely unknown backgrounds. 

We also saw that signed likelihood ratios are an excellent choice for the test statistic in deficit hawks. Other statistics perform similarly when all backgrounds are unknown (see section \ref{sec:stat_choice}), but clearly worse when discriminating known backgrounds is also relevant (see section \ref{sec:known_bg}). The adjusted $Q_\alpha$ and joint model methods can compensate for some, but not all of this difference.

\subsection{Further research}
\label{sec:future_research}
Interested researchers will find many opportunities for continued study. 
In particular, practical tests of the extensions to detection claims and multiple dimensions in sections \ref{sec:detections} and \ref{sec:multidim} may yield useful recommendations and refinements.

We did not consider inference on multiple parameters of interest, or using nuisance parameters beyond core regions used for discovery searches. Nuisance parameters can make likelihood ratios non-countlike -- for example, if background rates or shapes are floating, adding events in a background-rich region can increase the background expectation in the signal region, and make observations more excess-like. But perhaps some nuisance parameters are allowed under some conditions?

We considered the region set $R$ as fixed, but more complicated region choice algorithms may be possible too, as long as the deficit hawk remains countlike. For example, we could test only intervals with $\leq b$ observed events (e.g. $b=0$ tests only gaps), but we cannot test only intervals with $> b$ events. The latter is not countlike: unknown backgrounds might add events to regions and make them available for testing. This can lower the overall statistic, strengthen upper limits, and destroy coverage.
Similarly, it seems difficult to propose or filter candidate regions with heuristics such as alternate (faster) statistics.

All $Q_\alpha$ estimation in this work was done by Monte Carlo, which can be costly or inaccurate for high $\mu(s)$ (though see appendix \ref{sec:compute}). Although technological progress will continually expand the practical range for Monte Carlo, exact results or even analytical approximations \cite{asymptotic, nature_review} would still be useful. Some strategies for reducing the number of necessary Monte Carlo trials \cite{gross_vitells, tohm} could perhaps be generalized to deficit hawks. More generally, there should be a close analogy between deficit hawks and the literature for look-elsewhere effects and multiple hypothesis tests.

\section*{Acknowledgements}
I would like to thank Sara Algeri and Steven Yellin for useful discussions and feedback. I gratefully acknowledge support from a Kavli Fellowship granted by the Kavli Institute for Particle Astrophysics and Cosmology.

\bibliographystyle{apsrev4-1}
\bibliography{references}

\clearpage

\appendix

\section{Properties of the signed likelihood ratio}
\label{sec:proofs}

This appendix proves some claims about the signed likelihood ratio $t$ (eq.~\ref{eq:signed_plr}) in the text. In particular, we will show that $t$ is \emph{countlike}, as defined in section \ref{sec:stat_choice}, and can thus be used in deficit hawks.

We assume all notation introduced in section \ref{sec:notation}. For simplicity, `increasing' and `decreasing' are considered non-strictly, i.e.~$f(x)$ is  increasing if $x > y$ implies $f(x) \geq f(y)$ (not $f(x) > f(y)$). %

\begin{proposition}
\label{thm:changedl}
Consider a change to the log likelihood $\ell = \ln L$:
\begin{equation}
\label{eq:changedl}
\ell \rightarrow \tilde{\ell}(s) = \ell(s) + \Delta(s) 
\end{equation}
so that $\tilde{\ell}$ still has a maximum, now at $\tildehat{s}$. If $\Delta(s)$ increases with $s$, then the signed likelihood ratio $t$ increases under the change, i.e.~$\tilde{t}(s) \geq t(s) \; \forall s$. Conversely, if $\Delta(s)$ decreases with $s$, $t$ decreases 
under the change.
\end{proposition}
\begin{proof}

Using eq.~\ref{eq:changedl} for $\Delta$,
\begin{equation}
    \label{eq:leq}
    \Delta(\tildehat{s}) - \Delta(\hat{s}) = \big[ \tilde{\ell}(\tildehat{s}) -  \tilde{\ell}(\hat{s}) \big] - \big[ \ell(\tildehat{s}) - \ell(\hat{s}) \big].
\end{equation}
Since $\tildehat{s}$ maximizes $\tilde{\ell}$ and $\hat{s}$ maximizes $\ell$, the first term in brackets is positive, and the second negative. Thus,
\begin{equation}
\label{eq:lineq}
\ell(\tildehat{s}) - \ell(\hat{s}) \leq 0 \leq \tilde{\ell}(\tildehat{s}) -  \tilde{\ell}(\hat{s}) \leq \Delta(\tildehat{s}) - \Delta(\hat{s})
\end{equation}

Now consider the first part of the proposition, and assume $\Delta(s)$ increases with $s$. From eq.~\ref{eq:lineq}, we see $\tildehat{s} \geq \hat{s}$, i.e.~\emph{the best-fit increases}. To show $\tilde{t}(s) \geq t(s)$ for all $s$, we must thus consider three cases:
\begin{enumerate}[label=\Alph*]
    \item $s < \hat{s} < \tildehat{s}$, %
    \item $\hat{s} \leq s \leq \tildehat{s}$, %
    \item $\hat{s} < \tildehat{s} < s$, %
\end{enumerate}

Case B is straightforward, since $s$ lies on different sides of the sign flip in $t$ and $\tilde{t}$. That is, $t(s) \leq 0$ since $s \geq \hat{s} $ (and $u(s) \geq 0$), and $\tilde{t}(s) \geq 0$ since $s \leq \tildehat{s}$.

For case A and C, the sign function gives the same result for both $t(s)$ and $\tilde{t}(s)$, so we have
\begin{align}
    \frac{\tilde{t}(s) - t(s)}{\pm2} 
    &= \tilde{\ell}(s) - \tilde{\ell}(\tildehat{s}) - \ell(s) + \ell(\hat{s})  \nonumber \\ 
    &= \Delta(s) - \tilde{\ell}(\tildehat{s}) + \ell(\hat{s})  \label{eq:pre_expansion} \\
    &= \big[ \Delta(s) - \Delta(\hat{s}) \big] - \big[ \tilde{\ell}(\tildehat{s}) - \tilde{\ell}(\hat{s}) \big] \label{eq:expansion_a} \\
    &= \big[ \Delta(s) - \Delta(\tildehat{s}) \big] -  \big[ \ell(\tildehat{s}) - \ell(\hat{s}) \big] \label{eq:expansion_c}.
\end{align}
where the $-$ sign applies in case A and the $+$ sign in case C.
Note eq.~\ref{eq:expansion_a} and eq.~\ref{eq:expansion_c} are equivalent expansions of eq.~\ref{eq:pre_expansion}; alternatively each follows from the other by eq.~\ref{eq:leq}.

For case A, consider eq.~\ref{eq:expansion_a}. The first term in brackets is \emph{negative} as $\Delta$ is an increasing function (and $s \leq \hat{s}$ in case A), while the second term in brackets is \emph{positive} by eq.~\ref{eq:lineq}. Thus the whole expression is \emph{negative}, implying the desired $\tilde{t}(s) \geq t(s)$.

For case C, consider eq.~\ref{eq:expansion_c}. The first term in brackets is \emph{positive} as $\Delta$ is an increasing function (and $\tildehat{s} \leq s$ in case C), and the second term in brackets is \emph{negative} by \ref{eq:lineq}. Thus the whole expression is \emph{positive}, and again $\tilde{t}(s) \geq t(s)$ as desired.

Finally, consider the second part of the proposition, with $\Delta(s)$ decreasing in $s$. From eq.~\ref{eq:lineq}, we now see $\tildehat{s} \leq \hat{s}$, i.e. the best-fit \emph{decreases}. 
To show $\tilde{t}(s) \leq t(s)$ for all $s$, the proof again splits in three cases,
\begin{enumerate}[label=\Alph*]
    \item $s < \tildehat{s} < \hat{s}$,
    \item $\tildehat{s} \leq s \leq \hat{s}$,
    \item $\tildehat{s} < \hat{s} < s$.
\end{enumerate}
Case B is simple as before. Case A follows because eq.~\ref{eq:expansion_c} is now clearly positive under A's condition (rather than that of C), and case C because eq.~\ref{eq:expansion_a} is now clearly negative under C's condition (rather than that of A), by similar reasoning as above.
\end{proof}

Next, we use this proposition to show $t$ is countlike, provided its parameter $s$ increases the expected events everywhere. For full generality, we will assume $L$ is a \emph{binned likelihood}, since an unbinned likelihood (eq.~\ref{eq:unbinned_l}) is equivalent to a binned likelihood with infinitely many infinitesimal bins. That is,
\begin{align}
\label{eq:binnedl}
L &= \mathrm{Pr} \prod_{\text{bins} \, j} \text{Poisson}(n_{j}|\mu_{j}) \nonumber \\
\ln L &= \ln \mathrm{Pr} - \mu + \sum \limits_{\text{bins} \, j} n_j \ln \mu_j ,
\end{align}
where $n_{j}$ and $\mu_{j}$ are the observed and expected number of events in the bin numbered $j$; $\mu$ is the total expected number of events; and $\text{Poisson}(n|\mu) = \mu^n e^{-\mu} / n!$, the Poisson probability mass function.
We omitted model-independent constants from $\ln L$, since these cancel in likelihood ratios. All of $L$, $\mathrm{Pr}$, $\mu$ and $\mu_j$ are generally functions of $s$. If the likelihood is binned, regions used in the deficit hawk should always end at bin boundaries.

\begin{proposition}
The signed likelihood ratio $t(s)$ is countlike if $\mu_r'(s) \geq 0$ for any region $r \subset X$.
\end{proposition}
\begin{proof}
This is a simple corollary of proposition \ref{thm:changedl} above. 
First, we must prove $t$ increases if we add an event. Say the event falls in bin $j$. If bin $j$ is beyond the region $r$, $t$ is unchanged. If bin $j$ is inside $r$, then
\begin{align*}
\label{eq:adding_delta}
\Delta(s) &= \ln \mu_j.
\end{align*}
Differentiating, $\Delta'(s) = \mu'_j(s)/\mu_j(s) \geq 0$ by the assumption in the proposition. Thus, by proposition \ref{thm:changedl}, $\tilde{t} \geq t$, as required. Next, we must prove that growing $r$ by adding parts of space without observed events \emph{decreases} $t$. This may add several bins $j$ to the likelihood, each of which is empty and thus induces a change
\begin{equation*}
\Delta(s) = - \mu_j(s) .
\end{equation*}
Differentiating, $\Delta'(s) = - \mu'_j(s) \leq 0$ by the assumption in the proposition. Thus, by proposition \ref{thm:changedl}, $\tilde{t} \leq t$, as required. 
\end{proof}

\section{Computational considerations}
\label{sec:compute}

This appendix discusses practical aspects of using deficit hawks in settings with limited computational power. 

\subsection{Testing infinite region sets}
The main text mentioned several infinite region sets such as `all intervals' or `all intervals starting at 0'. In practice, as noted earlier in \cite{optitv}, we need only scan intervals ending on observed events or dimension boundaries: 
\begin{proposition}
\label{claim:event_bounded}
For $H[I,T]$, with $I$ the set of all intervals and $T$ a countlike statistic, the lowest-scoring interval $r^*$ has an observed event or a dimension endpoint at each of its boundaries.
\end{proposition}
\begin{proof}
Assume the contrary, the we could expand $r^*$, adding space with no observed events but some expected events. Clearly this new interval is still in $I$, but because $T$ is countlike, it will score even lower than $r^*$, a contradiction.
\end{proof}

Thus, for example, testing `all intervals' implies testing $(N + 2)(N+1)/2$ intervals. Similarly, testing `all intervals including the left endpoint' only involves testing $N+1$ intervals.

If we use $t_0$ (or $P_n$, $\chi$, etc.) as $T$, we can use an additional property:
\begin{proposition}
\label{claim:largest_sparsest}
For $H[I, t_0]$, no interval has both more expected and fewer observed events than $r^*$.
\end{proposition}
\begin{proof}
From differentiating eq.~\ref{eq:t0}, we see that $t_0$ always decreases as $\mu$ grows or as $N$ shrinks. Thus, if an interval has more expected and fewer observed events than $r^*$, it must have a lower $t_0$, a contradiction.
\end{proof}

Unfortunately, this does not result in a large computational benefit.
Even if we only need to compute the full test statistic on $N + 1$ intervals (the largest gap, largest containing 1 event, etc), we still need to determine $N$ and $\mu$ in all the $(N + 2) (N+1)/2$ intervals to find this subset. Computing $t_0$ is not that expensive -- unlike $t$, but unfortunately proposition \ref{claim:largest_sparsest} does not hold for $t$. For example, an interval with a high known background expectation may give a lower $t$ than a larger interval with an equal or smaller number of observed events (but a smaller known background expectation).

\subsection{Estimating and reusing $Q_\alpha(s)$}
The main computational challenge is generally not computing $T$ for the observed data, but estimating $Q_{\alpha}(s)$ by Monte Carlo simulations.
Still, this is often tractable. Unlike for computing discovery significances, we do not need to estimate extremely low quantiles for upper limits -- e.g.~we may only need $\alpha=0.1$, not $\alpha=2.8\cdot 10^{-7}$ (for `five sigma' discoveries) -- so very large simulations are not needed.
Moreover, analysts could first do a rough simulation at a wide range of $s$, then add more simulation statistics post-hoc to refine the $Q_{\alpha}$ estimate near the observed upper limit, if they require the latter to high accuracy.

As an example, most of the computations for this paper were done on the author's laptop, a Dell XPS 15 from 2020. For $H[I, t_0]$, it took just over a minute to simulate 1000 toys per point at every $s$ on a 152-point grid from $0.1$ to $1200$. Most figures used more statistics to get smoother plots, but this is enough to estimate upper limits reasonably well.
The full likelihood ratio $t$ is more expensive; even if we test $S_5$ and run only a few hundred toys per point, similar simulations took several CPU-hours.

Yellin pointed out that a $Q_\alpha(s)$ computed on one problem can often be re-used on another, if it is related by a simple change of coordinates \cite{optitv}. For example, most one-dimensional experiments can transform coordinates so that the signal is uniform in $[0,1]$. If there are no \emph{known} backgrounds, then $Q_\alpha$ depends only on $s$ (and the region choice and test statistic), not the signal model's PDF. Thus, we needed only one $Q_\alpha(s)$ curve for each region set to make figure \ref{fig:regioncomp_plaw}.

If there \emph{are} known backgrounds, $Q_\alpha$ is usually unique to a problem. Figure \ref{fig:statscomp_known} was the most computationally expensive plot in this work (and the only one requiring cluster computing), since every gridpoint in that plot has a different known background, and therefore required estimating new $Q_\alpha(s)$ curves for most of the statistics shown.

The joint model method is an exception to this, but only for some region choices. 
For example, if we test all intervals, there is for each $s$ a transformation that makes the summed signal-background model uniform on $[0,1]$. In that space, we still test all intervals. Thus, $Q_\alpha(s)$ is a shifted version of the background-free $Q_\alpha(s)$: the curves match on models with the same total expected events.
However, the regions in $S_5$ are bounded by quintiles of the \emph{signal} model, not the summed signal-background model. Thus, if we test $S_5$, we cannot relate $Q_\alpha$ to that of a background-free experiment anymore.

\subsection{High statistics}
Some analyses have many more events than the $\mu(s) \lesssim 100-1000$ events that we considered for examples in this paper. Deficit hawks may still be perfectly usable in such cases, for several reasons.

First, a single analysis does not need to test the performance under hundreds of different background scenarios, as this work did. Thus, it can estimate $Q_\alpha(s)$ to much higher values within a similar computational expense.

Second, when testing only a fixed region set such as $S_5$ or hand-picked cut options, the number of intervals to test does not grow with $s$.
Changing from an ordinary likelihood analysis (which tests one region) to one that tests e.g.~three options for an important cut increases the computational cost by at most a factor three -- assuming the experiment is already using simulations to estimate or verify $Q_\alpha(s)$ curves. Thus, even if testing `all intervals' or another infinite region set is infeasible, deficit hawks can still be used.

Third, and perhaps most importantly, just because an experiment detects many events does \emph{not} mean it needs to know $Q_\alpha(s)$ to values with very high $\mu_\text{sig}(s)$. Most of the observed events should be clearly unlike the signal -- otherwise deficit hawks are not very beneficial anyway.
As long as the experiment is be powerful enough to exclude modestly sized signals, the computational requirements are tractable.
A large unknown background does not make $Q_\alpha(s)$ any harder to estimate.

Experiments that can only exclude large signals are unlikely to be limited by statistics.
Even drastic measures such as discarding 90\% or 99\% of the data may not significantly weaken their results. %
Extrapolating $Q_\alpha(s)$ is perhaps a more palatable solution. As shown in figure \ref{fig:q_alpha}, $Q_\alpha(s)$ seems to converge to an asymptote for fixed region sets such as $S_4$ or $S_{10}$. For some infinite region sets, $Q_\alpha(s)$ also changes only slowly with $s$ at high $\mu(s)$.

\section{Additional background scenarios}
\label{sec:extra_scen}

This appendix shows more examples of one-dimensional region choices, and tests them against different unknown background shapes than in section \ref{sec:region_choice}. Apart from the background shapes, the tests here use the same setup as in section \ref{sec:region_choice} and figure \ref{fig:regioncomp_plaw}. In particular, we choose coordinates so that the signal is flat in $x\in [0,1]$. 

Many experiments search for a peak-like signal on top of a smoother background, in a quantity such as energy. After transforming coordinates, the background will have a U-shape, perhaps such as shown in the insets of figure \ref{fig:extra_halfpipe} or \ref{fig:extra_halfpipe_asym}. The more slowly the background changes with energy, the more symmetric the U-shape will be.

A sensible region set for these situations would be all intervals including some central point: perhaps the median signal energy, or the energy at which the signal has maximum amplitude. Figures \ref{fig:extra_halfpipe} and \ref{fig:extra_halfpipe_asym} shows the sensitivity of this region set in green. The performance is closer to simply testing all intervals than in figure \ref{fig:regioncomp_plaw}, where tailored region choices performed more significantly better.

In these examples, even testing $r^*$ alone would not improve much over testing all intervals, so there is little to be gained from more refined region choices. For example, analysts expecting a flat or slowly-varying background could test only \emph{symmetric} intervals around the central point, i.e.~intervals with the same expected signal event fraction on either side of the central point. Figures \ref{fig:extra_halfpipe} and \ref{fig:extra_halfpipe_asym} show the performance of this in dashed green -- as expected, the difference is small.

Restricting the region set has only small benefits in the examples above, because the unknown component is firmly present everywhere. This causes limits to converge to a fraction of the unknown strength (set by the minimum of the unknown PDF) as the unknown strength increases, leveling out any differences between region sets (for sets flexible enough to allow $r^*$). 
In figure \ref{fig:regioncomp_plaw}, we saw that refined region choices yield the most benefit for large unknown components, but in the examples above, the leveling force of the convergence sets in earlier.

We can also see this in the examples of figure \ref{fig:extra_halfbad} and \ref{fig:extra_staircase}. In figure \ref{fig:extra_halfbad}, half the measured space is background-free, and restricting region choices works very well: for strong unknown backgrounds, testing all intervals is three times worse than testing $r^*$ (the right half of the space).
In figure \ref{fig:extra_staircase}, the unknown component is staircase-shaped (as in the test in figure 3b of \cite{optitv}), and the advantage of restricting regions is again modest to small. %

\clearpage

\begin{figure}[t]
    \centering
    \includegraphics[width=\columnwidth]{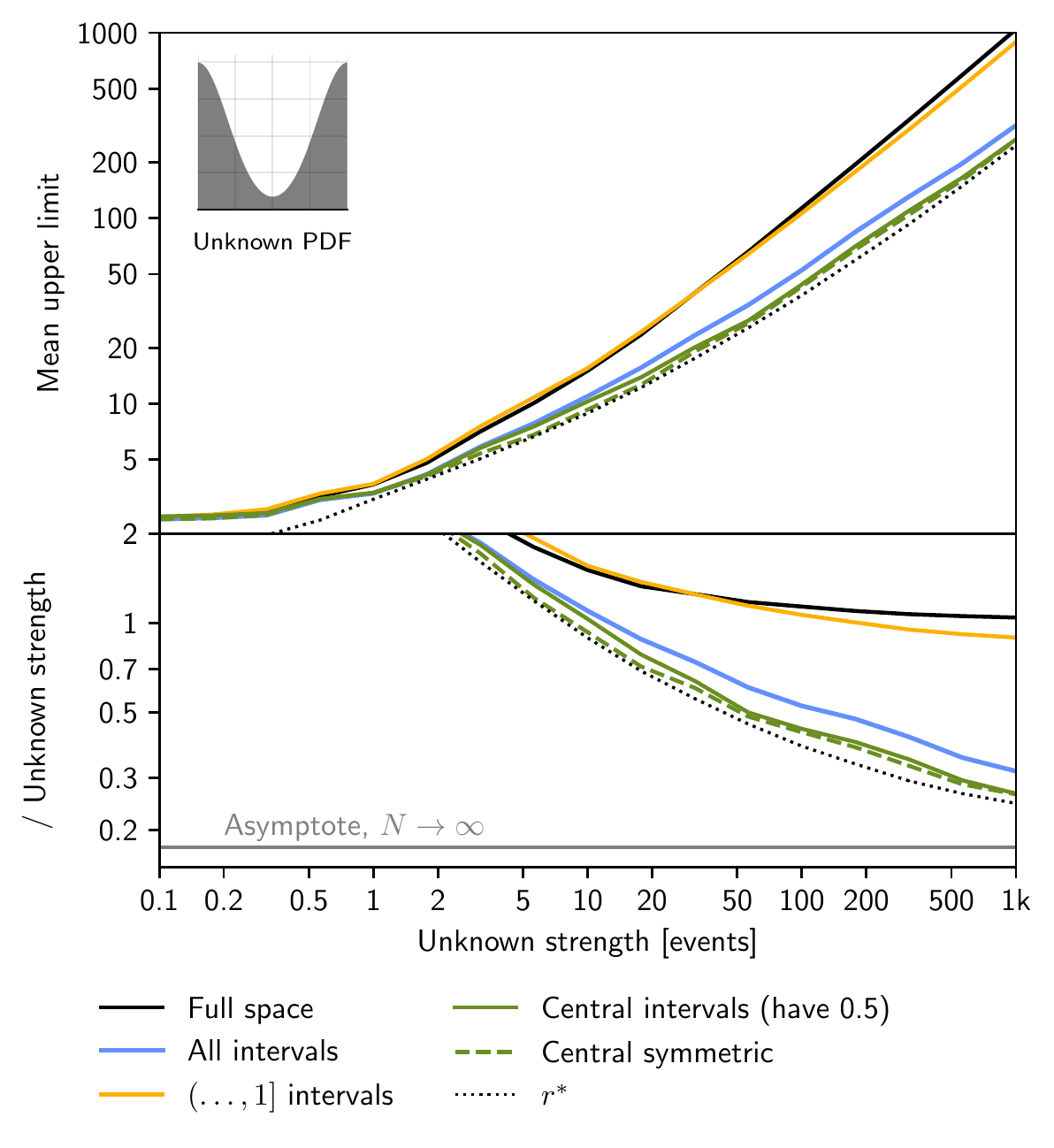}
    \caption{Mean 90\% confidence level upper limits on the expected event counts, for a U-shaped unknown component of different strengths.}
    \label{fig:extra_halfpipe}
\end{figure}

\begin{figure}[b]
    \centering
    \includegraphics[width=\columnwidth]{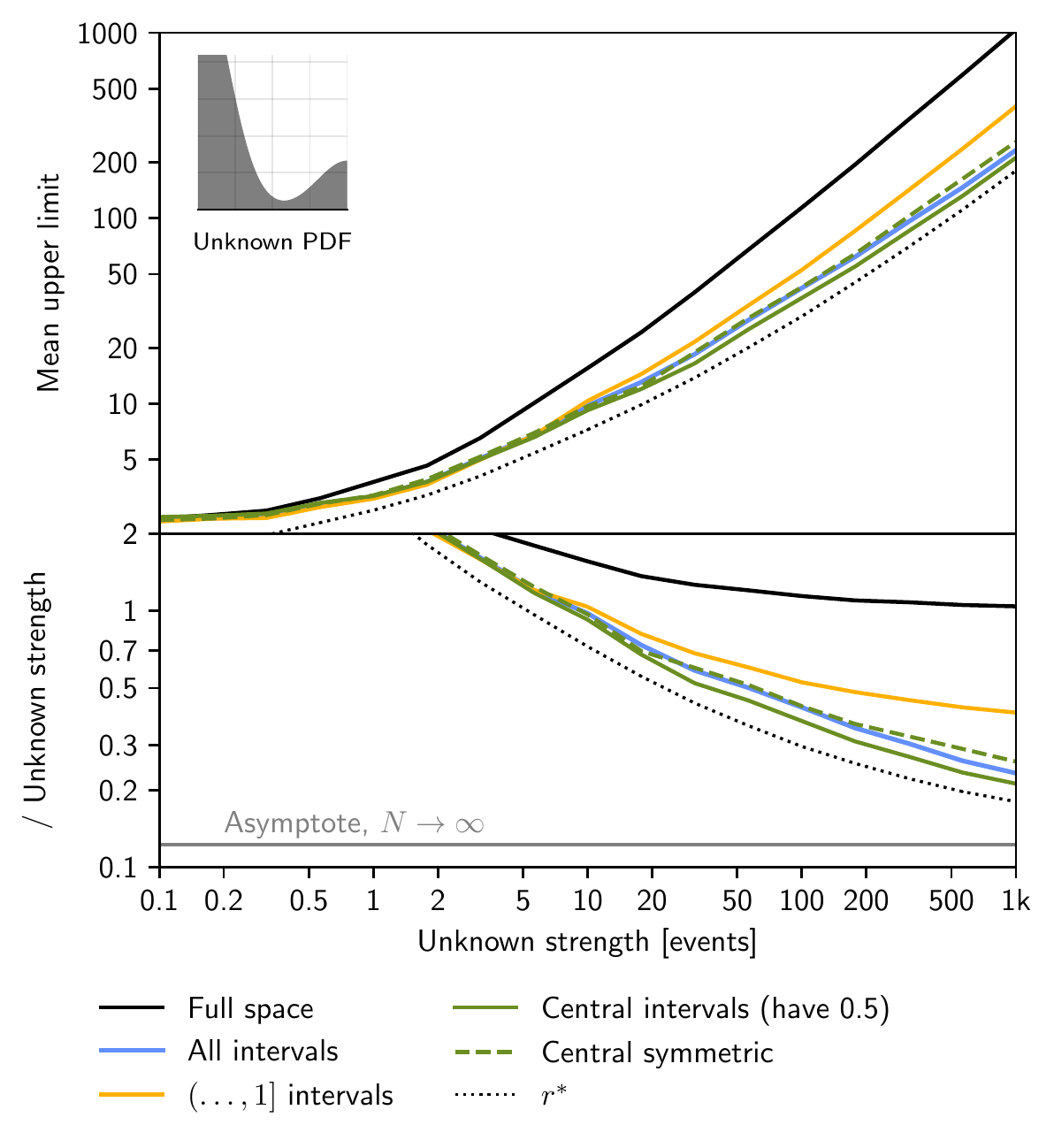}
    \caption{Mean upper limits for an asymmetric U-shaped unknown component.}
    \label{fig:extra_halfpipe_asym}
\end{figure}

\begin{figure}[b]
    \centering
    \includegraphics[width=\columnwidth]{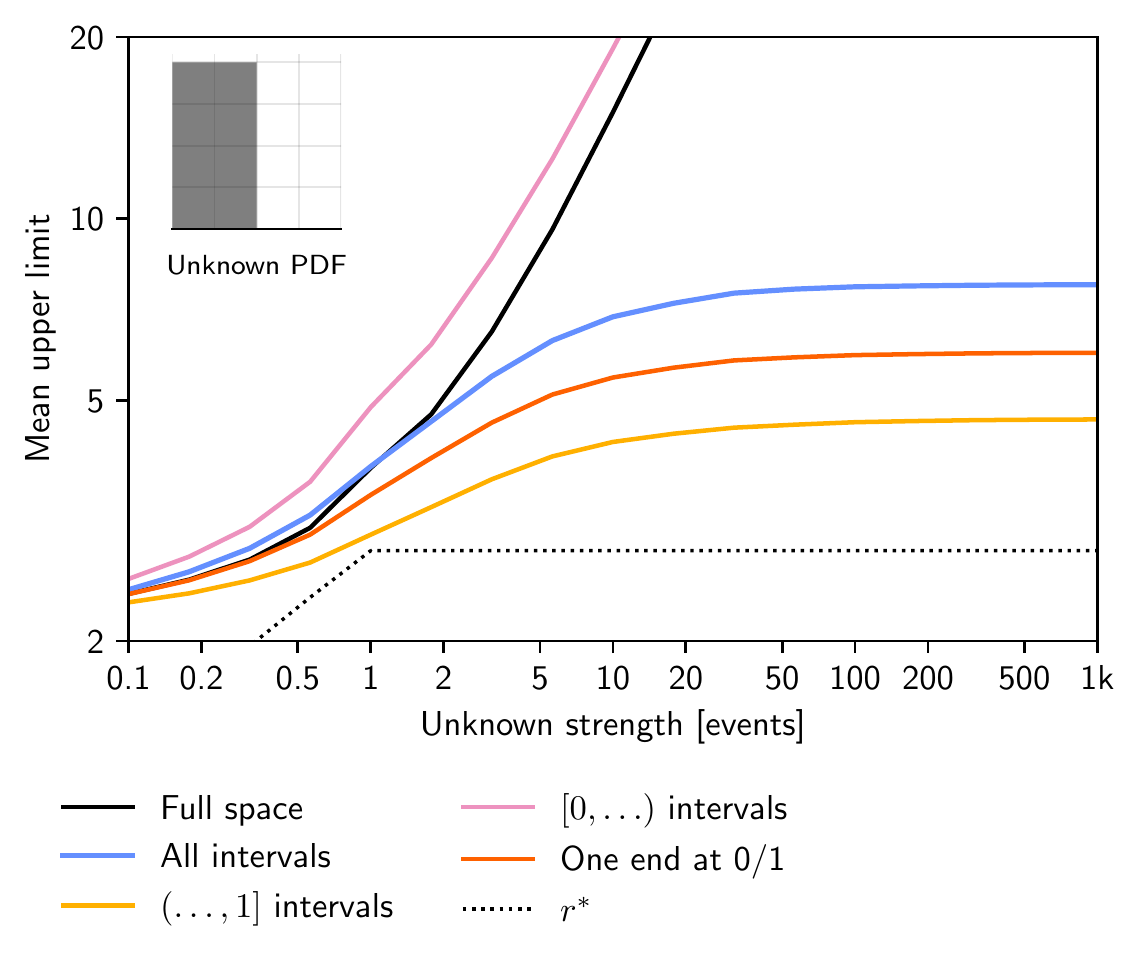}
    \caption{Mean upper limits for different region choices, for an unknown background in half the space.}
    \label{fig:extra_halfbad}
\end{figure}

\begin{figure}[b]
    \centering
    \includegraphics[width=\columnwidth]{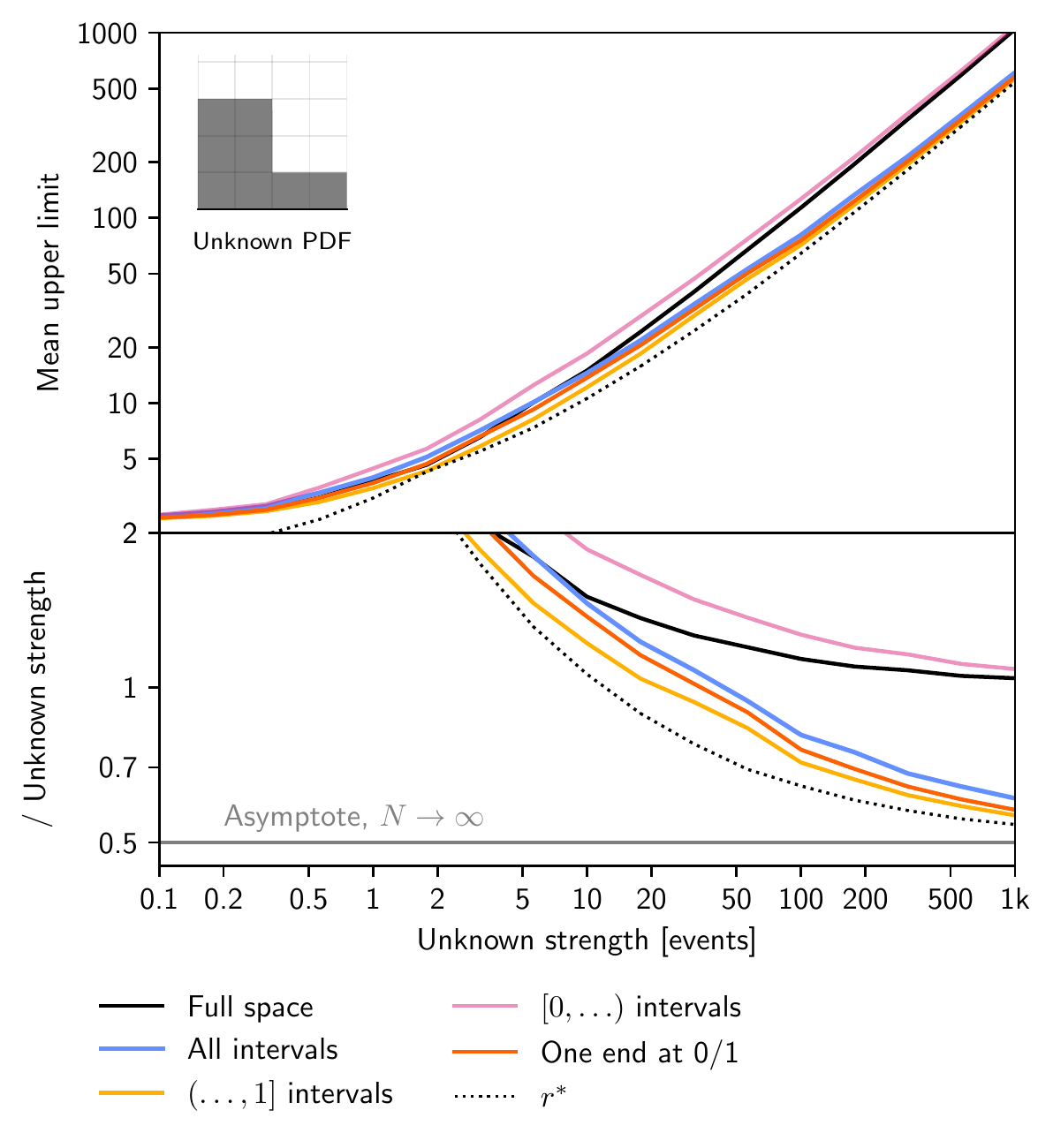}
    \caption{Mean upper limits for different region choices, for a staircase-shaped unknown component.}
    \label{fig:extra_staircase}
\end{figure}

\clearpage

\end{document}